\documentclass{article}

\usepackage[english]{babel} 

\usepackage[noend]{algpseudocode}
\usepackage[algo2e,ruled,linesnumbered,boxed,boxruled,vlined]{algorithm2e}
\usepackage{fullpage}
\usepackage{amssymb}
\usepackage{amsmath}
\usepackage{amsthm}
\usepackage{mathtools}
\usepackage{bm}
\usepackage[mathscr]{euscript}
\usepackage{comment}
\usepackage{url}
\usepackage{todonotes}
\usepackage{xcolor}  \usepackage[
    colorlinks=true,
    linkcolor=blue,      citecolor=teal,     urlcolor=magenta     ]{hyperref}
\usepackage{xspace}
\usepackage{xfrac}
\usepackage{faktor}
\usepackage{xparse}
\usepackage{enumerate}
\usepackage{bm}
\usepackage{array}
\usepackage{makecell}
\usepackage{pgfplots}
\pgfplotsset{compat=1.18}
\usepackage{cleveref}
\usepackage{caption}

\newtheorem*{maintheorem*}{Theorem}

\newtheorem*{otherlemma*}{Lemma}

\newtheorem{theorem}{Theorem}[section]
\newtheorem{definition}[theorem]{Definition}
\newtheorem{lemma}[theorem]{Lemma}
\newtheorem{corollary}[theorem]{Corollary}

\newtheorem{remark}[theorem]{Remark}

\crefname{lemma}{lemma}{lemmas}
\Crefname{lemma}{Lemma}{Lemmas}

\usepackage{mleftright}
\usepackage{xparse}

\DeclareFontFamily{OT1}{pzc}{}
\DeclareFontShape{OT1}{pzc}{m}{it}{<-> s * [1.10] pzcmi7t}{}
\DeclareMathAlphabet{\mathpzc}{OT1}{pzc}{m}{it}

\newcommand{\CC}{\ensuremath{\mathbb{C}}\xspace}

\newcommand{\NN}{\ensuremath{\mathbb{N}}\xspace}
\newcommand{\OO}{\ensuremath{\mathbb{O}}\xspace}

\newcommand{\QQ}{\ensuremath{\mathbb{Q}}\xspace}
\newcommand{\RR}{\ensuremath{\mathbb{R}}\xspace}

\newcommand{\ZZ}{\ensuremath{\mathbb{Z}}\xspace}

\newcommand{\sfb}{\mathsf{b}}

\newcommand{\bsz}{\ensuremath{\mathpzc{h}}\xspace}

\newcommand{\res}{\mathtt{Res}}

\DeclareDocumentCommand\OB{ g }{{ \mathcal{O}_B \IfNoValueF {#1} { \mleft( #1 \mright) } } }

\DeclareDocumentCommand\OO{ g }{{\mathcal{O} \IfNoValueF {#1} { \mleft( #1 \mright) } } }

\DeclareDocumentCommand\sOB{ g }{ {\widetilde{\mathcal{O}}_B \IfNoValueF {#1} { \mleft( #1 \mright) } } }

\DeclareDocumentCommand\sOO{ g } {{\widetilde{\mathcal{O}} \IfNoValueF {#1} { \mleft( #1 \mright) } } }

\DeclareDocumentCommand\Set{ m g }{{\mleft\{ #1\IfNoValueF {#2} { \;\middle\vert\; #2 }\mright\}}}

\DeclareDocumentCommand\Norm{ m g }{{\mleft\| #1 \mright\| \IfNoValueF {#2} { _{#2}}}}
\newcommand{\Normi}[1]{\Norm{#1}_{\infty}}

\DeclareDocumentCommand\norm{ m g }{{\| #1 \| \IfNoValueF {#2} { _{#2}}}}
\newcommand{\normi}[1]{\norm{#1}_{\infty}}
\newcommand{\normo}[1]{\norm{#1}_{1}}
\newcommand{\onenorm}[1]{\norm{#1}_{1}}
\newcommand{\normt}[1]{\norm{#1}_{2}}

\newcommand{\ceil}[1]{\mleft\lceil {#1} \mright\rceil}

\newcommand{\abs}[1]{\mleft| {#1} \mright|}

\renewcommand{\imath}{\mathtt{i}}
\newcommand{\lc}{\mathtt{lc}}
\newcommand{\tc}{\mathtt{tc}}
\newcommand{\sgn}{\mathtt{sgn}}

\newcommand{\Ae}{A_{\varepsilon}}
\newcommand{\wAe}{\widetilde{A}_{\varepsilon}}
\newcommand{\Aphi}{A_{\phi}}
\newcommand{\Aphie}{A_{\phi,\eps}}

\newcommand{\Pe}{P_{\varepsilon}}
\newcommand{\wPe}{\widetilde{P}_{\varepsilon}}
\newcommand{\Qe}{Q_{\varepsilon}}
\newcommand{\wQe}{\widetilde{Q}_{\varepsilon}}

\newcommand{\wt}[1]{\widetilde{#1}\xspace}

\newcommand{\aei}{\alpha_{\varepsilon,i}}

\newcommand{\waei}{\widetilde{\alpha}_{\varepsilon,i}}

\newcommand{\waej}{\widetilde{\alpha}_{\varepsilon,j}}
\newcommand{\wgej}{\widetilde{\gamma}_{\varepsilon,j}}
\newcommand{\wdej}{\widetilde{\delta}_{\varepsilon,j}}

\newcommand{\waed}{a_{\eps,d}}

\newcommand{\hp}{\widetilde{p}}
\newcommand{\hz}{\widetilde{z}}
\newcommand{\hatm}{\widetilde{m}}
\newcommand{\hP}{\widetilde{P}}

\newcommand{\usos}{\texttt{uSOS}\xspace}

\newcommand{\eps}{\varepsilon}

\newcommand{\cq}[1]{$(\mathsf{Q}_{{\ref{#1}}})$}

\newcommand{\cC}{\mathcal{C}}
\newcommand{\cF}{\mathcal{F}}
\newcommand{\cG}{\mathcal{G}}

\newcommand{\cX}{\mathcal{X}}

\newcommand{\lin}{\mathrm{lin}}
\newcommand{\maple}{\textsc{maple}\xspace}

\newcommand{\rP}{\mathscr{P}}
\newcommand{\rQ}{\mathscr{Q}}

\begin{document}

\title{Positive Univariate Polynomials: \\ 
SOS certificates, algorithms, bit complexity, and T-systems}
\author{Mat\'ias Bender~\thanks{Inria Saclay and CMAP, \'Ecole Polytechnique, IPP, France}
\and Philipp di Dio~\thanks{University of Konstanz, Germany}
\and Elias Tsigaridas~\thanks{Inria Paris and IMJ-PRG, Sorbonne Universit\'e, France}
}
\date{\today} 

\maketitle

\begin{abstract}

We consider certificates of positivity for univariate polynomials with rational coefficients that are positive over (an interval of)~$\mathbb{R}$.
Such certificates take the form of weighted sums of squares (SOS) of polynomials with rational coefficients.

We build on the algorithm of Chevillard, Harrison, Jolde{\c{s}}, and Lauter~\cite{chml-usos-alg-11}, 
and we introduce a variant that we refer to as \usos. Given a polynomial of degree~$d$ with maximum coefficient bitsize~$\tau$, we show that \usos computes a rational weighted SOS representation in $\widetilde{\mathcal{O}}_B(d^3 + d^2 \tau)$ bit operations; the resulting certificate of posivitity involves rationals of bitsize $\widetilde{\mathcal{O}}(d^2 \tau)$. This improves the best-known complexity bounds by a factor of~$d$ and completes previous analyses. We also extend these results to certificates of positivity over arbitrary rational intervals, via a simple transformation. In this case as well, our techniques yield a factor-$d$ improvement in the complexity bounds.

	Along the same line, for univariate polynomials with rational coefficients, we introduce a new class of certificates, which we call \emph{perturbed SOS certificates}.
They consist of a sum of two rational squares that approximates the input polynomial closely enough so that nonnegativity of the approximation implies the nonnegativity of the original polynomial.
	This computation has the same bit complexity and yields certificates of the same bitsize as in the weighted SOS case.

We further investigate structural properties of these SOS decompositions.
Relying on the classical result that any nonnegative univariate real polynomial is the sum of two squares of real polynomials, we show that the summands form an interlacing pair. Consequently, their real roots correspond to the Karlin points of the original polynomial on~$\mathbb{R}$, establishing a new connection with the T-systems studied by Karlin~\cite{Karlin-repr-pos-63}. This connection enables us to compute such decompositions explicitly. Previously, only existential results were known for T-systems. We obtain analogous results for positivity over $(0, \infty)$, and hence over arbitrary real intervals.

Finally, we present our open-source Maple implementation of the \usos algorithm, together with experiments on various data sets demonstrating the efficiency of our approach.

\end{abstract}

\newpage	

\tableofcontents

\newpage

\section{Introduction}

A univariate polynomial $A$, with real coefficients that takes only nonnegative values over $\RR$,
admits a decomposition as a sum of (two) squares of real polynomials;
this representation \emph{certifies} the nonnegativity of~$A$. 
Following Powers \cite{powers-certif-book}, a \emph{certificate of positivity}
is an algebraic identity\footnote{There are cases where the certificate might consists of several 
algebraic identities. This is not the case for the problem we consider, so we 
do not explore further this direction.}
that straightforwardly implies  the
nonnegativity of a (univariate) polynomial.
Our focus is mainly on univariate polynomials with rational coefficients.
Hence, in the analysis of the certificates, in addition to the involved polynomials, we also take into account the maximum number of bits we need to represent their coefficients. Along the same lines, the complexity of the corresponding algorithms refers to the number of bit operations.

The certificate(s) of positivity raises
the following mathematical, algorithmic, and complexity-related questions:
\begin{enumerate}[($\mathsf{Q}_1$)]\addtocounter{enumi}{-1}
\item \label{Q0} \label{cq:theorem}
    What is the mathematical framework, usually a theorem, that implies an algebraic identity (or more than one) corresponding to the certificate?
\item \label{Q1}  \label{cq:algo}  Is there an algorithm to compute the certificate and what is its (bit) complexity?
\item \label{Q2} \label{cq:cert-sz} What is the (bit)size of the certificate?
\item \label{Q3} \label{cq:verif} What is the (bit) complexity of verifying the algebraic identity induced by the certificate, that is, to verify that the certificate is correct?
\item \label{Q4} \label{cq:witness} If there is no certificate, 
that is, if $A$ could be negative,
  then, can we compute a witness point such that the evaluation of $A$ at this point 
  	is negative? What is the cost of computing the witness point and what is its (bit)size?
\end{enumerate}

In the case of a polynomial $A$ with real coefficients of even degree, say $d = 2m$,
it is well known that $A$ is nonnegative over $\RR$
if and only if it is a sum of two squares of polynomials;
this answers \cq{cq:theorem}.
This is a special version of a more general mathematical foundation based on Karlin's description of nonnegative polynomials in T-systems \cite{Dio-T-systes-arxiv}; we will exploit this connection further in the sequel.
In particular, we have the following equivalence.
\begin{equation}
	\label{eq:pos-certificate-R}
	 A(x) \geq 0 \text{ for all } x \in \RR 
	\quad\Leftrightarrow\quad
	A = P^2 + Q^2,
\end{equation}
where $A,P,Q\in\RR[x]$.
The equivalence in \eqref{eq:pos-certificate-R} serves as a certificate of positivity for $A$. Regarding its size, we notice that the polynomials $P$ and $Q$ are of  degree at most $m = d/2$;
the degree of $A$ must be even.
Because \eqref{eq:pos-certificate-R} involves polynomials with real coefficients,
it is not relevant to discuss bounds on the bitsize of the certificate.
To compute the polynomials $P$ and $Q$, we
need to compute and manipulate the roots of $A$, e.g.,~\cite[Chapter~8]{powers-certif-book}. 
Hence, the complexity of the decomposition is dominated by
the root-finding algorithm we employ for this task; this requires $\sOO(d)$ arithmetic operations;
see e.g.,~\cite{pan-rootfinding-jsc} and references therein.
We can verify the algebraic identity of the certificate in \eqref{eq:pos-certificate-R},
$A = P^2 + Q^2$, and so we can answer \cq{cq:verif}, 
either deterministically with direct computations
or probabilistically, by evaluating the left- and right-hand side polynomials at random numbers, 
as in polynomial identity testing.
Finally, for \cq{cq:witness} it suffices to return a number, say $t
\in \RR$, such that $A(t) < 0$.
If $A$ can be negative, then
it has real roots of odd multiplicity, or it is everywhere negative.
For the former case, we can choose as a witness point a number $t$ lying to the left or to the right of a real root.
If $A$ is globally negative, then any number $t \in \RR$ suffices.

If $A$ has rational coefficients and we opt for an SOS decomposition
with polynomials having rational coefficients, then things are
somewhat more complicated and there is a (slightly) different certificate. 
Specifically, there is a weighted SOS representation, that is, a representation as sum of squares of polynomials with rational coefficients
multiplied by positive rational numbers. 
Namely, the certificate is: 
\begin{equation}
	\label{eq:pos-certificate-Q}
	A(x) \geq 0\ \text{for all}\ x \in \RR 
	\quad\Leftrightarrow\quad
	A = \sum_{j=1}^{\nu} w_j\cdot s_j^2,
\end{equation}
where $A, s_j \in \QQ[x]$, $w_j \in \QQ_{\geq 0}$ for all $j \in [\nu]$, and some $\nu\in\NN$.

Pourcet \cite{pourcet-5squares},  improving a previous result of Landau \cite{Landau-usos-06}, 
proved that only five or less squares are needed, $\nu \leq 5$; but his proof is not constructive.
We refer the reader to \cite{kmv-prouchet-23} for recent advances in this direction. 

At the cost of having more summands, there is the constructive approach by 
Schweighofer~\cite{Schw-sos-99}
that, roughly speaking, successively subtracts positive quadratic polynomials
from $A$.
There is also the algorithm by 
Chevillard, Harrison, Jolde{\c{s}},
and Lauter \cite{chml-usos-alg-11} 
that 
computes a weighted SOS decomposition 
at the expense of involving $\nu = d+3$ summands, 
where $d$ is the degree of the polynomial. 
This is the algorithm that we focus on; we call it \usos.
The crux of the algorithm is that it performs a sufficiently small perturbation
to $A$ and then approximates the complex roots of the perturbed polynomial.
We introduce a slight variant of this algorithm, that we also call \usos and we refer to Section~\ref{sec:wu-sos-presentation} for a detailed presentation.

If we assume that $A$ has integer coefficients, $A \in \ZZ[x]$,
and the maximum coefficient bitsize is $\tau$,
then Magron, Safey El Din, and Schweighofer \cite{mss-wusos-alg-19}
studied the bit complexity of \usos.
They demonstrated that \usos computes the certificate in \eqref{eq:pos-certificate-Q}
in 
$\sOB(d^4 + d^3 \tau)$ bit operations 
\cq{cq:algo}. 
They estimated the bitsize of the certificate, 
based on \usos, to be 
$\sOO(d^3 + d^2\tau)$  \cq{cq:cert-sz}.
In particular, they show that the certificate involves $\OO(d)$ coefficients and their bitsize is $\sOO(d^2 + d \tau)$.
Finally, they show how to verify the certificate in $\sOB(d^4 + d^3 \tau)$ bit operations \cq{cq:verif}.

Along the same lines, Krick, Mourrain, and Szanto \cite{kms-usos-22} studied a more general (semi-algebraic) problem. That is, they presented an algorithm to support the following equivalence:
	$A \in \ZZ[x]$ is nonnegative on all real roots of  $B \in \ZZ[x]$ iff $A$ is SOS (of polynomials with rational coefficients) modulo $B$. The authors do not provide
a bit complexity analysis, however, we refer the reader to \cite{bkm-eff-pos-24} 
for recent improvements and generalizations.

We emphasize that the certificates of positivity imposed by \eqref{eq:pos-certificate-R} or \eqref{eq:pos-certificate-Q} are not the only ones. 
 There are also certificates based on
Bernstein basis \cite{bcr-bern-pos-dcg}, (dual) certificates based on the dual cone of weighted sums of square
polynomials \cite{papp-dual-cert-bsz} or based on sums of circuits \cite{dressler2017positivstellensatz},
 just to mention a few alternatives; we do not proceed further in these directions.

Finally, let us also mention the relation of nonnegative polynomials to T-systems.
The theory of T-systems goes back a long time and it is highly developed.
We refer the interested reader to \cite{KarStu-Tsystems-66,Dio-T-systes-arxiv} for a detailed exposition.
For our purposes, Karlin's work \cite{Karlin-repr-pos-63}, 
especially the following consequence \cite[Cor.\ 1]{Karlin-repr-pos-63}, is of utmost importance: Let $A \in\RR[x]$ with $\deg(A)= 2m$.
Then, the following are equivalent:
\begin{enumerate}[(i)]
    \item $A(x) >0$ for all  $x \in \RR$.
    \item There exist unique constants $\alpha,\beta>0$ and unique points $x_1,\dots,x_m,y_1,\dots, y_{m-1}\in\RR$,
    that we call \emph{Karlin points} of $A$ over $\RR$, with
\begin{equation}
    \label{eq:K-xi-yj}
	x_1 < y_1 < \dots < y_{m-1} < x_m,	
    \end{equation}
such that
\[A(x) = \alpha\cdot (x-x_1)^2\cdots (x-x_m)^2 + \beta\cdot (x-y_1)^2\cdots (x-y_{m-1})^2.\]
\end{enumerate}

Clearly, $\alpha$ is the leading coefficient of $A$.
There are variants for $A(x)>0$ or $A(x)\geq 0$ for all $x$ in $[a,b]$ or $[0,\infty)$; we refer the reader to \cite{KarStu-Tsystems-66,Dio-T-systes-arxiv}
for a comprehensive treatment. 
We also refer to \Cref{sec:T-systems} for a brief introduction and some additional details. 
Unfortunately, the non-constructive approach of Karlin does not provide us with an algorithm to compute the unique Karlin points $x_i$ and $y_j$.
Hence, T-systems were mainly of theoretical interests. 
We present a constructive approach in \Cref{sec:interlacing}.

\paragraph{Our contribution.}
We revisit and slightly modify the algorithm by Chevillard, Harrison, Jolde{\c{s}},
and Lauter \cite{chml-usos-alg-11} 
for decomposing a polynomial $A \in \ZZ[x]$, of degree $d$ and bitsize $\tau$, that is positive over $\RR$,
as an SOS of at most $d + 3$ polynomials with rational coefficients;
we call the variant \usos.
We complete and improve by a factor of $d$ the complexity analysis
by  Magron, Safey El Din, and Schweighofer~\cite{mss-wusos-alg-19}.

The main idea of \usos consists in perturbing the input polynomial $A$.
The perturbed polynomial is $\Ae := A - \eps\, M$, where $M = \sum_{k=0}^{d/2} x^{2k}$ is a positive polynomial and $\eps$ is the perturbation.
Our choice of a small enough $\eps = 2^{-\sfb}$ is
such that two things happen:
(i) $\Ae$ is positive, this is the main requirement of \cite{chml-usos-alg-11} (Lem.~\ref{lem:eps-value-pos}),
and
(ii) the roots of $\Ae$ are close to the roots of $A$ (Lem.~\ref{lem:eps-value-sep});
this is a new requirement.
The second requirement guarantees that the separation bounds,
that is the minimum distance between two distinct roots,
of $\Ae$ and $A$ are almost the same, 
even though their bitsizes are different (Lem.~\ref{lem:eps-value}). 
Also, the asymptotic values of $\eps$ for both requirements are the same. 
Notice that the appropriate value of $\eps$ does not depend
on the actual polynomial $A$ but only on its degree, $d$
and bitsize $\tau$.

This choice of $\eps$ allows to deduce that the bit complexity of approximating, in sufficient precision,
the roots of $\Ae$ is asymptotically the same as the complexity of approximating the roots of $A$;
this saves us a factor of $d$ for this step of the algorithm (Cor.~\ref{cor:Ae-roots-approx-complexity}). 
Then, we exploit the fan-in algorithm from approximate multipoint evaluation \cite{pt-struct-mat-j}
to compute good approximate SOS decomposition for $\Ae$ and thus for $A$
(Cor.~\ref{cor:P-Q-complexity} and Lem.~\ref{lem:kappa-value}).
In this way, we obtain a bound $\sOB(d^2 \tau)$ 
for the complexity of computing a certificate of positivity for $A$ (\Cref{thm:usos-bit-complexity}). 
This improves the previously known 
bound by a factor of $d$ and answers~\cq{cq:algo}.
Our complexity bound matches the bitsize of the certificate
and hence we improve the complexity of the algorithms supporting the answers to \cq{cq:cert-sz} (Lem.~\ref{lem:bsz-of-certif}) and \cq{cq:verif} (\Cref{thm:certification-bit-complexity}).

The following theorem summarizes our results on representing a positive polynomial
as a weighted sum of squares of polynomials. 
\begin{maintheorem*}[\usos and weighted SOS representation]
	Let $A \in \ZZ[x]$ be a square-free polynomial of degree $d = 2m$ and maximum coefficient
	bitsize $\tau$. 
	If $A$ is positive over $\RR$, then there is an algorithm to compute a weighted SOS representation of $A$
	as 
	\[
	A(x) = \sum\nolimits_{j=1}^{\nu} w_j \, s_j^2(x),
	\] 
	where $w_j \in \QQ_{>0}$, $s_j \in \QQ[x]$, and $\nu \leq d + 3$,
	at the cost of $\sOB(d^3 + d^2\tau)$ bit operations.
	The bitsize of $w_j$'s and the coefficients of $s_j$'s is at most $\sOO(d \tau)$.
\end{maintheorem*}

Based on the previous theorem and by applying a transformation from Chevillard et al~\cite{chml-usos-alg-11},
we also provide bit complexity estimates for certificates and algorithms for
the positivity of a polynomial over any interval (\Cref{thm:wpusos-interval}).
If the bitsize of the endpoints of the interval is $\sigma$, then the bit complexity of the algorithm is
$\sOB(d^3 + d^2\tau + d^2\sigma)$. To achieve this bound, it is not enough to apply the algorithm
supported by the previous theorem directly; this will give us an extra $d$ factor in the complexity bound.
We save this factor by studying how the separation bound of $A$ changes after the transformation, see Sec.~\ref{sec:pos-over-interval} for details.

To demonstrate the efficiency of \usos we present 
an open-source prototype implementation in \maple 
and experiments on various data sets (\Cref{sec:implementation}).
The experiments (\Cref{sec:experiments}) verify the bounds on the bitsize of the certificates
and demonstrate the efficiency of the algorithm.

An important and challenging ingredient of an implementation of \usos 
is the computation of the perturbation $\eps := 2^{-\sfb}$.
Other approaches compute $\eps$ by relying on the minimum of $A$
or by repeatedly dividing by 2.
These choices do not result in efficient implementations
because they force us to compute with numbers having the worst case bitsize right from the beginning and/or the number of steps they perform 
depends on the bitsize of $\eps$, that is $\sOO( \sfb) = \sOO( d \tau)$; the latter leads to an overall complexity 
that is quadratic in the bitsize of the input, that is $\sOB( \tau^2)$.
A natural question to ask is if we can come up with an efficient implementation to compute $\eps$ that does not dominate or alter the overall worst case bit complexity of the algorithm. 
We compute $\eps$ by performing exponential binary search to obtain $\sfb$;
this requires $\OO( \lg(\sfb)) = \OO( \lg( d \tau))$ steps. In this way, 
we obtain enormous speedups in the running times of our implementation of \usos,
while at the same time guarantee the best theoretical worst case bit complexity bounds.  
We refer to Section~\ref{sec:implementation-details} for further details
and an experimental evaluation.

\medskip

\noindent
\emph{Perturbed SOS certificates.}
Although the \usos algorithm provides efficient nonnegativity certificates, 
it requires to approximate the roots of 
the auxiliary polynomial $\Ae$ with precision $\OO(d\tau)$ bits 
and produces certificates consisting of $\OO(d)$ SOS summands. 
To overcome these limitations, for polynomials with rational coefficients, 
we introduce \emph{perturbed SOS certificates}, which certify nonnegativity of a polynomial $A \in \ZZ[x]$ by constructing a rational SOS approximation $B$ that is sufficiently close to $A$ in the sup-norm. 
Specifically,  
if $\normi{A-B} < 2^{-\sfb}$ with $\sfb > \sOO(d \tau)$, 
then the nonnegativity of $B$ implies the nonnegativity of $A$
(Thm.~\ref{thm:approxImpliesNonNeg}). 
For the case of square-free positive polynomials of even degree, 
the certificate has the simple form $B := \wt{P}^2 + \wt{Q}^2$, 
where the polynomials $\wt{P}, \wt{Q} \in \QQ[x]$ are obtained from rational approximations of the roots of $A$. 
The bit complexity of computing such a certificate is $\sOB(d^3 + d^2\tau)$, while the bitsize of $\wt{P}, \wt{Q}$ is bounded by $\sOO(d \tau)$
(Thm.~\ref{thm:pert-sos-cert-complexity}). This establishes perturbed SOS certificates as a refinement of \usos that yields minimal-size representations while retaining fully explicit complexity and precision guarantees.
\bigskip

\noindent
\emph{T-systems and effective Karlin points.}
Finally, we study structural properties of the SOS representation (Sec.~\ref{sec:interlacing}).
We establish
a mathematical and algorithmic connection 
of the certificates of positivity of univariate polynomials
with the T-systems, introduced by Karlin~\cite{KarStu-Tsystems-66}. 
First, we consider the case a real polynomial, $A \in \RR[x]$, that is positive over $\RR$ (Sec.~\ref{sec:Karlin-R}); hence it admits a representation as $A = P^2 + Q^2$.
The polynomials $P$ and $Q$ are interlacing (Lem.~\ref{lem:PQ-interlace}),
as a direct consequence of Hermite-Bielher theorem (Thm.~\ref{thm:HB})
and thus their roots are the unique points, 
that we call Karlin points, Eq.~\eqref{eq:K-xi-yj},
appearing in the decomposition induced by the T-systems of Karlin (Cor.~\ref{cor:Karlin-pts-R}).
Even though, this seems to be a fundamental property, we were not able 
to find it in the literature.
Besides on relying on Hermite-Bielher theorem~\cite{rs-athp-2002}
for the interlacing property, we give an alternative proof (Sec.\ref{sec:alternative-interlacing})
based on algebraic manipulations of the roots of the polynomials.
We refer to Fisk's survey~\cite{Fisk-interlace-bk} for a thorough study of real rooted and interlacing polynomials.
The following theorem summarizes our results:

\begin{maintheorem*}[Positivity over $\RR$, interlacing, and Karlin points]
	If $A  =  \sum_{k=0}^d a_k x^k \in \RR[x]$, of degree $d = 2m$,
	is square-free and positive over $\RR$, then 
	$A(x) = a_d \, P(x)^2 + a_d \, Q(x)^2$, where $P, Q \in \RR[x]$
	are interlacing of degrees $m$ and $m-1$, respectively.
	Moreover, the real roots of $P$ and $Q$ are the Karlin points of $A$ over $\RR$.
\end{maintheorem*}

We also present variants of the previous theorem for positivity over $(0, \infty)$, or any interval of $\RR$,
we discover the corresponding interlacing polynomials
and we show that their real roots are the Karlin points of $A$ in the interval of interest;
see Sec.~\ref{sec:Karlin-interval}.

\paragraph{Organization}

First, we present a detailed description of the various steps 
of the \usos algorithm by Chevillard, Harrison, Jolde{\c{s}},
and Lauter \cite{chml-usos-alg-11} (\Cref{sec:wu-sos-presentation}).
Then, \Cref{sec:wpu-sos-steps-complexity} studies 
the bit complexity of the steps and the bitsize of the various quantities 
involved in the computations.
In \Cref{sec:overall-complexity-wpu-sos} we present the overall complexity of the algorithm and the certificate,
and in \Cref{sec:pos-over-interval} we consider nonnegativity over an interval. 
\Cref{sec:implementation} presents our implementation and experiments. 
In \Cref{sec:interlacing} we establish the connection of positive real polynomials, SOS decompositions,
T-systems, and interlacing polynomials.
For a brief background on T-systems we refer to \Cref{sec:T-systems} and \cite{Karlin-repr-pos-63,Dio-T-systes-arxiv}.

Finally, in the Appendix,
we give and alternative proof of Lemma~\ref{lem:PQ-interlace} (Sec.~\ref{sec:alternative-interlacing}) 
and  we present auxiliary results 
on separation bounds and root approximation of univariate polynomials (Sec.~\ref{sec:prelim-root-approx}),
bounds on the minimum of a polynomial and an 
approximation variant of the 
fan-in algorithm from multipoint evaluation (Sec.~\ref{sec:fan-in}).
We also present a bird's eye view of T-systems (\Cref{sec:T-systems}).

\paragraph{Notation}

We denote by $\OO$, resp. $\OB$, the arithmetic, respectively bit,
complexity; we also use $\sOO$, resp. $\sOB$, to ignore (poly-)logarithmic factors.
For a polynomial $A = \sum_{k=0}^d a_k x^k \in \CC[x]$ of degree $d$
we denote by $\onenorm{A}$ resp.\ $\normi{A}$, the one resp.\ the infinity, norm of the vector $(a_0,\dots,a_d)$.
We denote by $\lc(A) = a_d$, resp. $\tc(A) = a_0$, the leading, resp. tailing, coefficient of $A$.
If $A \in \QQ[x]$, 
then the bitsize of $A$ is the maximum bitsize of
its coefficients, including a bit for the sign. 
For $a \in \QQ$, its bitsize 
is the maximum bitsize of the numerator and the denominator.
If $A \in \QQ[x]$ has degree $d$ and bitsize $\tau$, then we also say 
that $A$ has size $(d, \tau)$. 
We write $\Delta_{\alpha}(A)$ or just  $\Delta_{\alpha}$ to denote 
the minimum distance between a root $\alpha$ of $A$ and 
any other root. We call this quantity {\em local separation bound}.
We also write $\Delta_i$ instead of $\Delta_{\alpha_i}$.
Also $\Delta(A)=\min_{\alpha}{\Delta_{\alpha}(A)}$ or just $\Delta$ denotes
the {\em separation bound}, that is the minimum distance between all
the roots of $A$.
Finally, let $D(c, r) = \{ x \in \CC\,:\, \abs{x-c} \leq r \}$.
Given a complex number $z \in \CC$ such that $\abs{z} < 2^\tau$, we say that $\wt{z}$ is an approximation up to an absolute precision $\lambda \in \NN$, if 
$\abs{z - \wt{z}} < 2^{-\lambda}$. 
Then, 
the bitsize of the approximation is at most 
$\tau + 2 \lambda + 2$ and we can represent it 
as a dyadic fraction of the form $a 2^{-b}$, 
for $a \in \ZZ$ and $b \in \NN$.

We should note the constants in the various bounds we present are not the best possible. A more detailed analysis can improve them. We decided to present them, 
even in their rough form, to demonstrate that there are no hidden non-constant factors in the $\OO$ notation of the bounds.

\section{The \usos algorithm}
\label{sec:wu-sos}

Consider the following polynomial that has even degree $d = 2m$ and is positive over $\RR$:
\[A(x) = \sum_{k=0}^d a_k\cdot x^k = a_d\cdot \prod_{i=1}^d (x - \alpha_i) .\]
Our goal is to provide a representation of $A$ as a weighted sum of squares of polynomials. 

We assume that $A$ is square-free and has no real roots, that is  
$\alpha_i \not\in \RR$, for all $i \in [d]$. 
These assumptions are without loss of generality; we refer to Sec.~\ref{sec:assumptions}
for a detailed discussion.

To simplify various calculations in the sequel, 
we also need to require the leading coefficient 
of $A$ to be such that
\[ \tfrac{1}{2} \leq \lc(A) = a_d \leq 1 \enspace,\]
and that all the other coefficients of $A$ are rational numbers of the same denominator
and of bitsize bounded by $\tau$.
Hence, when the input is a polynomial $A$
with integer coefficients, first, we multiply $A$
with a rational number in the interval $[1/(2a_d), 1/a_d]$ to ensure this condition.
This operation (or requirement)
does not change neither the positivity of $A$ nor the complexity bound, 
hence we will assume it in our analysis.

We present in detail (Sec.~\ref{sec:wu-sos-presentation})
 the various steps of the \usos algorithm 
by Chevillard, Harrison, Jolde{\c{s}}, and Lauter \cite{chml-usos-alg-11}
that decomposes $A$
as a weighted sum of squares of polynomials with rational coefficients.
The presentation leads to the 
precise bit complexity analysis of Sec.~\ref{sec:overall-complexity-wpu-sos}.

\subsection{A detailed presentation of the \usos\ algorithm}
\label{sec:wu-sos-presentation}

\noindent
\textbf{Input:} A polynomial 
$A = \sum_{k=0}^d a_k x^k = a_d \prod_{i=1}^d (x - \alpha_i) \in \QQ[x]$
of even degree, $d = 2 m$.

\vspace{3pt}
\noindent
\emph{Assumptions: } (i) $A$ is square-free, (ii) $A$ is positive over $\RR$, and
(iii) The leading coefficient of $A$ is in $[\tfrac{1}{2}, 1]$,
while the other coefficients are rationals of bitsize at most $\tau$, having 
a common denominator.

\vspace{7pt}
\noindent
\textbf{Output:} A weighted SOS decomposition of $A$, that is a representation of the form 
\begin{equation}
	A(x) = \sum\nolimits_{j=1}^{\nu} w_j \, s_j^2(x) \enspace,
\end{equation}
where $w_j \in \QQ_{>0}$ and $s_j \in \QQ[x]$.
It holds $\nu \leq d + 3$.

\subsubsection*{[Step 1] Rewrite $A(x)$ using  $M(x)$ and $\varepsilon$ \cite[Sec.~5.2.2]{chml-usos-alg-11}}

Consider the polynomial \[M(x) = \sum\nolimits_{j=0}^{m} x^{2j},\] 
that is the sum of even powers $x^{2j}$ less than or equal to $d$.
Notice that $M(t) > 0$ for all $t \in \RR$.
Write $A$ as \[A(x) = A(x) - \varepsilon M(x) + \varepsilon M(x)
= \Ae(x) + \varepsilon M(x),\]
where $\eps > 0$.
In particular, $\eps$ should be small enough so that 
$\Ae := A - \eps M$ is strictly positive over $\RR$.
In Sec.~\ref{sec:epsSmallToSOS} we estimate a precise value for $\eps$.

\subsubsection*{[Step 2] Approximate the roots of $\Ae$ and compute $\wPe$ and $\wQe$ \cite[Sec.~5.2.3]{chml-usos-alg-11}}
 
The polynomial $\Ae$ is strictly positive and has no real roots;
let its factorization to linear factors be 
 \[\Ae(x)
  = \lc(\Ae)\cdot \prod_{i=1}^{d}(x - \aei),
  \quad\text{ where } \aei \in \CC\quad\text{and}\quad
  \waed \coloneq \lc(\Ae) > 0.\]

We can approximate the roots of $\Ae$, $\aei$, with rationals, up to any desired accuracy, 
say $2^{-\kappa}$, for some positive integer $\kappa$.
Let the approximations be \[\waej^{\pm} = \wgej \pm \imath \wdej,\] where $\wgej, \wdej \in \QQ$
and $j \in [m]$. Then, it holds 
\[ \abs{ \aei^{\pm} - \waei^{\pm} } \leq 2^{-\kappa}  \enspace .\]
In turn, the rational approximations of the roots lead to a polynomial $\wAe$
with rational coefficients, 
that is 
\[
	\wAe(x) = \waed \prod_{j=1}^{m}(x - \waej^{+})(x - \waej^{-})= 
	\waed  \prod_{j=1}^{m} (x - \wgej + \imath\, \wdej)
	\prod_{j=1}^{m} (x - \wgej - \imath\, \wdej),
\]	
where we can additionally assume that $\wdej \geq 0$, for all $j \in [m]$.
Moreover,
\[
\prod_{j=1}^{m} (x - \wgej + \imath\, \wdej) = 
\wPe(x) + \imath \wQe(x)
\quad \text{ and } \quad
\prod_{j=1}^{m} (x - \wgej - \imath\, \wdej) = 
\wPe(x) - \imath \wQe(x) ,
\]
which implies to the following representation of $\wAe$:
\[
	\wAe(x) = \waed \,
	\big(\wPe(x) + \imath \wQe(x) \big) \, \big( \wPe(x) - \imath \wQe(x) \big) = 
	 \waed \,  \big( \wPe(x)^2 + \wQe(x)^2 \big) .
\]
Therefore, as the numbers $\waej^{\pm} = \wgej \pm \imath \wdej$ 
approximate the roots  $\aei$,
we also deduce that \[\wAe = \waed (\wPe^2 + \wQe^2)\] approximates the polynomial $\Ae$.
The two polynomials $\Ae$ and $\wAe$ have the same leading coefficient.
Let their difference be
\[\Ae(x) - \wAe(x)  \eqcolon B(x) = \sum_{k=0}^{d-1} b_k x^k.\]
In this way, we obtain the following relation for $A(x)$: 
\begin{equation}
	\label{eq:A=P+Q+B+eM}
	A(x) = 
	\Ae(x) + \eps M(x) = \wAe(x) + B(x) + \eps M(x)
	= \waed \,\wPe(x)^2 + \waed\, \wQe(x)^2 + B(x) + \eps M(x) .
\end{equation}

\subsubsection*{[Step 3] Write $B(x) + \varepsilon M(x)$ as SOS \cite[Sec.5.2.5]{chml-usos-alg-11} }

As the first two summands of \eqref{eq:A=P+Q+B+eM} are weighted sum of squares, 
to represent $A$ as a weighted SOS, we should44 express $B(x) + \eps M(x)$ as a weighted SOS. 
For this, we exploit the identities 
\[
 x = (x + \tfrac{1}{2})^2 - (x^2 + \tfrac{1}{4}) 
\quad \text{ and } \quad
 -x = (x - \tfrac{1}{2})^2 - (x^2 + \tfrac{1}{4}) .
 \]
 In this way, for any $c > 0$, we have 
 \[
  \pm c\,x^{2k+1} = c \, \Bigl(x^{k+1} \pm \frac{x^{k}}{2} \Bigl)^2 
  - \,c\, \Bigl(x^{2k+2} + \frac{x^{2k}}{4} \Bigl).
\]
Using these identities, the odd-degree terms of 
$B(x) = \sum_{k=0}^{d-1} b_k x^k = \Ae(x) - \wAe(x)$ become
\[
b_{2k+1} x^{2k+1} = \abs{b_{2k+1}} 
\Bigl(x^{k+1} + \sgn(b_{2k+1}) \frac{x^{k}}{2} \Bigl)^2 
-\, \abs{b_{2k+1}} \Bigl(x^{2k+2} + \frac{x^{2k}}{4} \Bigl)
\]
and consequently
\begin{equation}
	\label{eq:BeM-remainder}
B(x) + \varepsilon\, M(x) =
 \sum_{k=0}^{m-1} \underbrace{\abs{b_{2k+1}}}_{w_k} \Bigl(x^{k+1} + \sgn(b_{2k+1}) \frac{x^{k}}{2}\Bigl)^2 
 + \sum_{k=0}^m \Bigl( \underbrace{\varepsilon + b_{2k} - \abs{b_{2k-1}} 
 	- \tfrac{1}{4} \abs{b_{2k+1}}}_{w_{m+k}} \Bigl) x^{2k},
\end{equation}
where by convention \(b_{-1} = b_{2m+1} = 0\).
If, for every $k$, it holds
\begin{equation}
	\label{eq:eps-ineq}
	\varepsilon \geq \frac{1}{4}|b_{2k+1}| - b_{2k} + |b_{2k-1}|, 
\end{equation}
then \eqref{eq:BeM-remainder} is a weighted SOS representation.
For the inequalities in \eqref{eq:eps-ineq} to hold, 
we should approximate the roots of 
$\Ae(x) = A(x) - \eps M(x)$ with enough precision, say $\kappa$,
so that the polynomial $\wAe$ is close to $\Ae$,
thus their different $B \coloneq \Ae - \wAe$ is small
and so the coefficients $b_k$ are small
(compared to $\eps$).

\subsection{The bit complexity of the various steps}
\label{sec:wpu-sos-steps-complexity}

We estimate the (bit)size of the various quantities appearing
in the process of the \usos and the complexity of the various 
operations.
Along the way we estimate the value of $\eps$ that suffices to perturb the original 
polynomial $A$ and the precision, $\kappa$, that we need to approximate the roots 
of the perturbed polynomial $\Ae$.
We express both as a function of $d$ and $\tau$.
Recall, that we assume that $A$ is a square-free polynomial of degree $d$, 
positive over $\RR$, and its leading coefficient is in $[\tfrac{1}{2}, 1]$,
while the other coefficients are rationals of bitsize at most $\tau$ having 
a common denominator.

\subsubsection*{[Step 1] Rewrite $A(x)$ using  $M(x)$ and $\eps$} \label{sec:epsSmallToSOS}

We estimate a suitable small value for $\eps = 2^{-\sfb}$
to ensure that the polynomial $\Ae(x) \coloneq A(x) - \eps M(x)$ is positive for every $x \in \RR$.
We assume \[0 < \eps \leq \tfrac{1}{8} \leq \tfrac{a_d}{4}.\] 
Then, the leading coefficient of $\Ae$ is $\lc(\Ae) = a_d - \eps > 0$.
Our analysis of this step follows closely \cite{mss-wusos-alg-19}.

An upper bound on the magnitude of the (real) roots of $A_{\epsilon}$ 
\cite[Theorem~1]{emt-dmm-j} is 
\begin{align}\label{eq:boundsizeroots}
1 \leq R := 2 \frac{\normi{\Ae}}{\abs{\lc(\Ae)}} \leq 2^{\tau + 5} .\end{align}
Then, for any $x \in \RR$ such that $\abs{x} \geq R$,
we have $\Ae(x) > 0$ (and $A(x) > 0$) as the leading coefficient of $\Ae$ 
is positive.
If we choose $\eps$ such that 
\[\eps = 2^{-\sfb} \leq 
\frac{\min_{\abs{x} \leq R} A(x)}{\max_{\abs{x} \leq R} M(x)},\]
then we ensure that $\Ae$ is positive for all $\abs{x} \leq R$.

Regarding $M(x)$, 0 is the only real root of its derivative, thus
$\max_{\abs{x} \leq R} M(x) = \max\Set{1, M(R)}  = M(R)$.
It holds 
\[
	\max_{\abs{x} \leq R} M(x) = M(R) =
	 \sum_{k=0}^{d/2} R^k \leq (\tfrac{d}{2} + 1) R^{d/2} 
	 \leq  (\tfrac{d}{2} + 1) \, 2^{d(\tau + 5)/2}.\]
The global minimum of $A$ is reached at a critical value, 
that is the evaluation of $A$ at a root of its derivative,
and, for obtaining worst case bounds, we can assume that is in the interval $(-R, R)$. So,
Lemma~\ref{lem:eval-A-at-da} implies 
\[2^{-4d \tau - 16d\lg{d}}
		\leq  \min_{\abs{x} \leq R} A(x) \leq 
		2^{ 2 d\tau + 8 d\lg{d}}.\]

Overall, to ensure that $\Ae$ is positive, we choose
$\eps = 2^{-\sfb}$, where
\begin{equation}
	\label{eq:b-step-1}
	\sfb \geq 4 d \tau + 8 d \lg{d} .
\end{equation}

\begin{lemma}
	\label{lem:eps-value-pos}
	If $\eps = 2^{-\sfb}$, with $\sfb \geq 4 d \tau + 8 d \lg{d}$, then 
	$\Ae(x) \coloneq A(x) - \eps M(x) >0$ for all $x \in \RR$.
\end{lemma}

\subsubsection*{[Step 2] Approximate the roots of $\Ae$ and compute $\wPe$ and $\wQe$}
\label{sec:same-sep}

We approximate the roots of  $\Ae$
and then, using the approximations, we construct the polynomials $\wPe$ and $\wQe$.
To approximate the roots of $\Ae$ with rationals up to any desired precision,
say $2^{-\kappa}$, for a positive integer $\kappa$,
first,  we need to isolate them and then approximate them to any desired precision.
We use well known algorithms for this task, 
the main ingredients of which are the 
splitting circle method and (variants of) the Newton operator, e.g.,~\cite{pan-rootfinding-jsc,msw-aprox-fact-15}
and references therein.
The complexity of the rootfinding algorithms (mainly) depends on the (aggregate) separation bound of the roots of $\Ae$; that is the minimum distance between the roots, e.g.,~\cite{emt-dmm-j} and references therein.
We choose an $\eps \coloneq 2^{-\sfb}$ small enough so that the separation bounds of $A$ and $\Ae$ are similar;
this allows to bound the complexity of approximating the roots of $\Ae$ in terms of the separation of $A$
and save a factor of $d$ in the overall complexity and the bitsize of the certificate.

We proceed as follows: first we compute the suitable value for $\eps$, 
then we bound the complexity of approximating the roots of $\Ae$ up to precision $\kappa$, 
and, finally, we estimate the complexity of computing $\wPe$ and $\wQe$ (as function of $d$, $\tau$, and $\kappa$).

\paragraph{The separation bound of $\Ae$.}
We relate the separation bound of $\Ae$ in terms of the separation bound of $A$, using Lemma~\ref{lem:approx-sep}.
In our case, $A$ plays the role of $p$
and $\Ae$ plays the role of $\hp$.
Notice that the leading coefficients of $A$ and $\Ae$ are different, 
they are $a_d$ and $a_d - \eps$, respectively.
Thus, to apply \Cref{lem:approx-sep} we need to consider the one norm of 
the following difference 
\[
\begin{aligned}
	\onenorm{\tfrac{a_d - \eps}{a_d} A - \Ae} & \leq \onenorm{\tfrac{a_d - \eps}{a_d} A - A + \eps M} 
	& (\Ae = A - \eps M)\\
	& \leq  \eps \onenorm{\tfrac{1}{a_d}A - M}  
	\leq \eps\,  2\, d \,\onenorm{A} &  ( \onenorm{M} \leq \tfrac{d}{2} + 1)\\
	 & = 16 \eps d \, \tfrac{1}{8} \onenorm{A} 
	\leq 16 \eps d \,\tfrac{a_d - \eps}{a_d} \onenorm{A}   & (\eps \leq \tfrac{a_d}{4})\\
	& \leq 2^{-\sfb +\lg{d} + 4}\cdot   \onenorm{\tfrac{a_d - \eps}{a_d} A}. & (\eps = 2^{-\sfb})
\end{aligned}
\]
Notice that the polynomials $A$ and $\tfrac{a_d - \eps}{a_d} A$ have the same root and separation bounds.
Then, Lem.~\ref{lem:approx-sep},
for sufficiently small $\eps$ or equivalently for sufficiently
big $\sfb$, implies that the separation bounds of $A$ and $\Ae$ are related
with small constant depending on the degree, see \eqref{eq:z-hz-sep}.
In particular,
we should choose a positive integer $\sfb$ 
that satisfies the following three conditions
that correspond to \cref{eq:sfb-cond-1,eq:sfb-cond-2,eq:sfb-cond-3} 
of Lem.~\ref{lem:approx-sep}:
 \begin{enumerate}[(i)]
 	\item  $\sfb \geq \lg{d} + 4 + \max\{ 8d, d \lg{d} \}$ and $\sfb$ is a power of two.
\item $\displaystyle 2^{(-\sfb + \lg{d} + 4)/2} \leq \frac{\Delta_i(A)}{2 d}$, or equivalently,  
	$\sfb \geq -2 \lg\Delta_i(A) + 3 \lg(d) + 6$. 
Using \cite[Theorem~1]{emt-dmm-j} we can bound $\Delta_i(A)$ to obtain 
$\sfb \geq 2 d \tau + 12 d \lg{d} + 8$.
\item  $\displaystyle 2^{-\sfb/2} \leq \frac{ \prod_{j \not= i}(\alpha_i - \alpha_j)}{16\cdot (d+1)\cdot 2^{\tau_A}\cdot M(\alpha_i)^d}$.
It holds $M(\alpha_i) = \max\{1, \abs{\alpha_i} \} \leq 2^{\tau + 3}$,
	using an upper bound on the roots of $A$, e.g.,~\cite[Theorem~1]{emt-dmm-j}.
Also using Lem.~\ref{lem:eval-dA-at-a}, we obtain
	\begin{equation}
		\label{eq:Pi-eval}	
		P_i \coloneq \prod_{j \not= i}\abs{\alpha_i - \alpha_j}
		= \tfrac{1}{d \, a_d} \abs{A'(\alpha_i) }
		\geq 2^{-3d \tau - 3d\lg{d} - \lg{d}} .
	\end{equation}

	These bounds lead to the inequality $\sfb \geq 5 d \tau + 9 d \lg{d} + 12$.
 \end{enumerate}

Therefore, if we choose $\sfb$ such that   
\begin{equation}
	\label{eq:b-sep}
	\sfb \geq 5 d \tau + 9 d \lg{d} + 12, 
\end{equation}
then (i), (ii), and (iii) are simultaneously satisfied.
In this case, by Lem.~\ref{lem:approx-sep}, Eq.~\eqref{eq:z-hz-sep}, for any $i \in [d]$, 
\begin{equation}
	\label{eq:sep-A-Ae}
	(1 - \tfrac{1}{d}) \, \Delta_i(A) 
	\leq 
	\Delta_i(\Ae) 	
	\leq 
	(1 + \tfrac{1}{d}) \, \Delta_i(A) .
\end{equation}

The previous discussion leads to the following lemma:

\begin{lemma}
	\label{lem:eps-value-sep}
	If $\eps = 2^{-\sfb}$, with $\sfb \geq 5 d \tau + 9 d \lg{d} + 12$, then \eqref{eq:sep-A-Ae} holds.
\end{lemma}

By combining Lemmata~\ref{lem:eps-value-pos} and \ref{lem:eps-value-sep}
and considering all the roots of $\Ae$, we have the following

\begin{lemma}[The value of $\eps$]
	\label{lem:eps-value}
	If $\eps = 2^{-\sfb}$, with $\sfb \geq 5 d \tau + 9 d \lg{d} + 12$, then 
	$\Ae(x) = A(x) - \eps M(x) >0$ for all $x \in \RR$, and 
	\begin{equation}
		\label{eq:ag-sep-A-Ae}
		(1 - \tfrac{1}{d})^d \, \prod_{i=1}^{d} \Delta_i(A) 
		\leq 
		\prod_{i=1}^{d} \Delta_i(\Ae) 
		\leq 
		(1 + \tfrac{1}{d})^d \, \prod_{i=1}^{d} \Delta_i(A) .
	\end{equation}
\end{lemma}

\paragraph{Bit complexity of approximating the roots of $\Ae$.}

To approximate the roots of $\Ae$, up to precision $2^{-\kappa}$, for a positive integer $\kappa$,
we employ the algorithm supported by \Cref{thm:root-isol-approx}
that computes complex numbers $\waej^{\pm} \in \QQ[\imath]$, where $j \in [m]$,
such that 
\[
  \norm{\Ae - \waed \prod\nolimits_{j=1}^{m}(x - \waej^{+})(x - \waej^{-})}
	\leq 2^{-\sfb} \norm{\Ae}.
\]
We let \[\wAe(x) \coloneq \waed \prod_{j=1}^{m}(x - \waej^{+})(x - \waej^{-}).\]
The rootfinding algorithm returns the real and imaginary part 
of the $\waej^{\pm} = \wgej \pm \imath \wdej$ 
as dyadic fractions of the form 
$a\, 2^{-b}$, where $a \in \ZZ$ and $b \in \NN$;
all fractions have the same denominator. 
The next lemma bounds the bit complexity of approximating the roots of $\Ae$,
both the real and imaginary parts, up to precision $2^{-\kappa}$,
as a function of $d$, $\tau$, and $\kappa$. We estimate the value of $\kappa$ in the next subsection.

\begin{corollary}
	\label{cor:Ae-roots-approx-complexity}
	If $\eps = 2^{-\sfb}$, with $\sfb \geq 5 d \tau + 9 d \lg{d} + 12$ (that is as in Lemma~\ref{lem:eps-value}),
	then one can compute rational approximations (of the real and imaginary part) of the roots of $\Ae$,
	up to precision $2^{-\kappa}$, for a positive integer $\kappa$,  
	at the cost of $\sOB(d^3 + d^2\tau + d \kappa)$ bit operations.
\end{corollary}
\begin{proof}
	We bound the various quantities appearing in the complexity bounds of \Cref{thm:root-isol-approx}.

	First, we estimate the bound on the coefficients, 
	$\tau_{\Ae} \leq \ceil{ \frac{\normi{\Ae}}{\lc(\Ae)}}  \leq \tau+5$.
	
	To bound the (aggregate) separation bound of the roots of $\Ae$, we employ Eq.~\eqref{eq:ag-sep-A-Ae}. Then,
	\[
	-\lg \prod_{i=1}^{d} \Delta_i(\Ae) 
	\leq 
	-\lg (1 - \tfrac{1}{d})^d \, \prod_{i=1}^{d} \Delta_i(A) 
	\leq
	2 -\lg \prod_{i=1}^{d} \Delta_i(A) = \sOO(d \tau),
	\]	
	where for the last equality we refer to \cite[Theorem~1]{emt-dmm-j}.
	
	It remains to bound 
	$\hP_i \coloneq \prod_{j \not= i}\abs{\waei - \waej} = \tfrac{1}{d \, a_{\eps,d}}\abs{\Ae'(\waei)}$.
	It holds $\tfrac{P_i}{2} \leq \hP_i \leq 2 P_i$ \cite[proof of Theorem~4]{msw-aprox-fact-15}
	and so 
	\[
		-\lg \hP_i \leq -\lg \tfrac{1}{2} - \lg{P_i} 
		\leq 3d \tau + 3d\lg{d} + \lg{d} = \sOO(d \tau) \enspace,
	\]
	where the last inequality is due to \eqref{eq:Pi-eval}.
	By combining all the bounds, we conclude the proof.
\end{proof}

\paragraph{The cost of constructing the polynomials $\wPe$ and $\wQe$.}

It remains to actually compute the polynomials $\wAe$, 
$\wPe$, and $\wQe$ from the approximations $\waej$'s, 
based on Lemma~\ref{lem:fan-in}.
It holds \[\abs{\waej - \aei} \leq 2^{-\kappa}.\]
The fan-in algorithm, supported by Lemma~\ref{lem:fan-in},
computes the polynomial
\[\wAe(x) = \lc(\Ae) \prod_{j=1}^{m}(x - \waej^{+})(x - \waej^{-}),\] 
and the polynomial $\wPe$ and $\wQe$ (and $\wAe$) from the product
\[
	\prod_{i=1}^{m} (x - \wgej + \imath\, \wdej) = 
	\wPe(x) + \imath \, \wQe(x).
	\]
Recall, that $\Ae$ is a positive polynomial, so it admits a representation $\Ae = \Pe^2 + \Qe^2$. As we approximate the roots of $\Ae$, we compute approximations $\wPe$ and $\wQe$,
so that, after simplifications,
\[ 
\normi{\Pe - \wPe} \leq 2^{-\kappa + (4m-4)(\tau+5) + 32m - (\lg{m}+5)^2 -7}
\leq 2^{-\kappa + 2d \tau + 26d - \lg^2{d} - 30},
\]
and similarly for $\Qe$ and $\wQe$.
The cost is  $\sOB(d(\kappa + d \tau))$ bit operations,
it holds $\normi{\wPe}, \normi{\wQe} \leq 2^{d \tau + 4d}$,
and the bitsize of all three polynomials 
is at most $\OO(\kappa + d \tau)$.

\begin{corollary}
	\label{cor:P-Q-complexity}
	Assume $\eps = 2^{-\sfb}$, with $\sfb \geq 5 d \tau + 9 d \lg{d} + 12$ (that is as in Lemma~\ref{lem:eps-value}).
If we are given a rational approximation of the roots of $\Ae$, with precision $2^{-\kappa}$, 
	for a positive integer $\kappa$, then 
	we can compute the polynomials $\wPe$ and $\wQe$
	at the cost of $\sOB(d^2 \tau + d \kappa)$ bit operations. 
	The bitsize of the polynomials is $\sOO(\kappa + d \tau)$.
\end{corollary}

\subsubsection*{[Step 3] Write $B(x) + \varepsilon M(x)$ as SOS}

Now, we have computed rational approximations of the roots of $\Ae$ such that 
$\abs{\aei - \waei} \leq 2^{-\kappa}$, that correspond to the polynomial $\wAe$.
Based on Lem.~\ref{lem:fan-in} we deduce that 
\[
	\normi{\Ae - \wAe} \leq 
	2^{-\kappa + 2d \tau + 26d - \lg^2{d} - 30},
\]
and, since $B := \Ae - \wAe$,  it holds 
\[\normi{B} \leq 2^{-\kappa + 2d \tau + 26d - \lg^2{d} - 30}.\]
To satisfy the inequality \eqref{eq:eps-ineq}, 
that is $\eps \geq |b_{2k+1}|/4 - b_{2k} + |b_{2k-1}|$, 
the following inequality needs to hold 
\begin{equation}
	\label{eq:kappa-ineq}
	\eps = 2^{-\sfb} \geq 2^{-\kappa + 2d \tau + 26d - \lg^2{d} - 30}
	\quad\Rightarrow\quad
	\kappa \geq 5d\tau + 40 d \lg{d} \enspace,
\end{equation}
where we also use the bound on $\sfb$ from Lemma~\ref{lem:eps-value}.
Consequently, $\eps$ and all the coefficients $b_k$ of $B$ have bitsize $\sOO(d \tau)$.
This leads to the following lemma:

\begin{lemma}
	\label{lem:kappa-value}
	Let $\eps = 2^{-\sfb}$, with $\sfb \geq 5 d \tau + 9 d \lg{d} + 12$ (Lem.~\ref{lem:eps-value}).
	Then, if $\kappa \geq 5d\tau + 40 d \lg{d}$, that is $\kappa = \sOO(d \tau)$, then 
	 \eqref{eq:eps-ineq} holds for all $k$.
	 The bitsize of the positive rationals $w_j$ in \eqref{eq:BeM-remainder} is also $\sOO(d \tau)$.
\end{lemma}

\subsection{Overall complexity estimates}
\label{sec:overall-complexity-wpu-sos}

The previous two sections imply that 
for a given polynomial $A \in \ZZ[x]$, of degree $d = 2m$ and bitsize $\tau$,
that is positive over $\RR$, 
\usos computes a representation of $A$ 
as a weighted SOS of polynomials with rational coefficients, as in \eqref{eq:pos-certificate-Q}.
In particular, it represents $A$ as  
\begin{equation}
	\label{eq:wsos-decomp}
	A(x) =  \waed \, \wPe(x)^2 + \waed\, \wQe (x)^2 + 
		\sum_{k=0}^{m-1} w_k \, \Bigl(x^{k+1} \pm \frac{x^{k}}{2}\Bigl)^2 
 + \sum_{k=0}^m w_{m+k} \, x^{2k},
\end{equation}
where $\waed = a_d - \eps \in \QQ_{\geq 0}$,
$\eps \in (0, \tfrac{1}{8})$,
$\wPe, \wQe \in \QQ[x]$,
and $w_j \in \QQ_{\geq 0}$, for $0 \leq j \leq d$.

\begin{theorem}[Bit complexity of \usos]
		\label{thm:usos-bit-complexity}
	Let $A \in \ZZ[x]$ be a square-free polynomial of degree $d = 2m$ and maximum coefficient
	bitsize $\tau$. 
	If $A$ is positive over $\RR$, then \usos computes a weighted SOS representation of $A$
	as in \eqref{eq:wsos-decomp}, see also  \eqref{eq:pos-certificate-Q}, 
	at the cost of $\sOB(d^3 + d^2\tau)$ bit operations.
\end{theorem} 
 
\begin{proof}
	The first step of the \usos involves computing the polynomial $\Ae$. This requires $\OO(d)$ additions 
	of numbers of bitsize $\sOO(d \tau)$ (Lemma~\ref{lem:eps-value}). 
	So, the bit complexity is $\sOB(d^2 \tau)$.
	
	The second step requires to approximate the roots of $\Ae$ up to precision $\kappa$. 
	As $\kappa = \OO(d \tau)$ (Lemma~\ref{lem:kappa-value}), this
	costs $\sOB(d^3 + d^2 \tau)$ bit operations (Cor.~\ref{cor:Ae-roots-approx-complexity}).
	This also includes the cost of computing the polynomials 
	$\wPe$ and $\wQe$~(Cor.~\ref{cor:P-Q-complexity}).
	
	The last step requires $\OO(d)$ additions of numbers of bitsize $\sOO(d \tau)$
	to construct the coefficients $w_j$.

	Hence, the overall bit complexity is $\sOB(d^3 + d^2 \tau)$.	
\end{proof}

\begin{lemma}[Bitsize of the certificate]
	\label{lem:bsz-of-certif}
	The representation in \eqref{eq:wsos-decomp} involves at most $d + 3$ summands.
	The bitsize of the rationals involved in the representation is at most $\sOO(d \tau)$,
	while their total bitsize  is  $\sOO(d^2\tau)$. 
\end{lemma}
\begin{proof}
	The number of summands follows directly from the algorithm \cite{chml-usos-alg-11} and \eqref{eq:wsos-decomp}.

	We notice that $\eps$, $\waed$ (Lemma~\ref{lem:eps-value}), 
	the coefficients of $\wPe$ and $\wQe$ (Cor.~\ref{cor:P-Q-complexity}),
	and the rationals $w_j$ (Lemma~\ref{lem:kappa-value})
	all have bitsize $\sOO(d \tau)$.
	
	The polynomials $\wPe$ and $\wQe$ have at most $d+2$ coefficients,
	so the total number of rational numbers in the right hand side of \eqref{eq:wsos-decomp} is at most 
	$2d + 4$. Thus, the bitsize of all the rational appearing in 
	the right hand side of \eqref{eq:wsos-decomp} is $\sOO(d^2 \tau)$.
\end{proof}

\begin{theorem}[Bit complexity of verifying the certificate]
	\label{thm:certification-bit-complexity}
	We can verify the positivity certificate in \eqref{eq:pos-certificate-Q}, that is the identity in \eqref{eq:wsos-decomp} using 
	$\sOO( d^2\tau)$ bit operations.
\end{theorem}
\begin{proof}
	The right hand side of \eqref{eq:wsos-decomp} involves 
	the squaring of two polynomials, that is, $\wPe$ and $\wQe$, that have 
	degree $d/2$ and bitsize $\sOO( d \tau)$.
	Each squaring corresponds to one polynomial multiplication
	that costs $\sOB(d^2 \tau)$ bit operations. 
	Then, it suffices to compare the coefficients of the left and right hand sides,
	which we can do in linear time. 
\end{proof}

\subsubsection{Dropping the assumptions}
\label{sec:assumptions}

To certify the nonnegativity of a univariate polynomial using a weighted SOS representation, 
it suffices to provide such a representation for square-free polynomials with no real roots.

To justify this, assume that $A$ is not square-free and consider its square-free factorization:
\[
A = \prod\nolimits_{\mu} A_{\mu}^{2 \mu} \, \prod\nolimits_{\nu} A_{\nu}^{2 \nu + 1}.
\]
Each factor raised to an even power, say $A_{\mu}^{2 \mu}$, satisfies $A_{\mu}(x)^{2 \mu} \geq 0$ for all $x \in \RR$, and thus does not affect the sign of $A$. Consequently, we may disregard such factors when certifying nonnegativity. For factors raised to odd powers, $A_{\nu}^{2 \nu + 1}$, it suffices to analyze the contribution of $A_{\nu}$ to the sign of $A$. Therefore, we may restrict attention to square-free polynomials.

Moreover, if $A$ has real roots and is nonnegative over~$\RR$, then all real roots must have even multiplicity. Hence, the real roots are roots of the even-powered factors $A_{\mu}$.

To construct a weighted SOS representation for a polynomial $A$ that is not square-free, we consider again its square-free factorization:
\[
A = \prod\nolimits_{\mu} A_{\mu}^{2 \mu} \, \prod\nolimits_{\nu} A_{\nu}^{2 \nu + 1}
= \underbrace{\prod\nolimits_{\mu} (A_{\mu}^{\mu})^2 \, \prod\nolimits_{\nu} (A_{\nu}^{\nu})^2}_{S^2} \, \prod\nolimits_{\nu} A_{\nu}
= S^2 \, \prod\nolimits_{\nu} A_{\nu}.
\]
If $A$ is nonnegative over~$\RR$, then each $A_{\nu}$ is a positive, square-free polynomial with no real roots. We may apply the \usos algorithm to each $A_{\nu}$ to obtain a representation of the form $A_{\nu} = \sum_{j_{\nu}} w_{j_{\nu}} s_{j_{\nu}}^2$. Multiplying by $S^2$ yields a weighted SOS representation for $A$.

If $A$ is not nonnegative, i.e., if there exists $t \in \RR$ such that $A(t) < 0$, then at least one of the $A_{\nu}(t)$ must be negative, and in fact, an odd number of them must satisfy $A_{\nu}(t) < 0$. Such a point $t$ necessarily lies between two real roots of the product $\prod_{\nu} A_{\nu}$.

\subsubsection{Witness point of non-nonnegativity}

What if $A$ is not nonnegative? In this case, $A$ has at least one real root.
By isolating the real roots of $A$ (that is, by computing intervals with rational endpoints, 
each containing exactly one real root) we can compute rational points between successive 
real roots. We refer to these as \emph{intermediate points}.

If $A$ is not nonnegative, then there exists at least one intermediate point 
$t \in \QQ$ such that $A(t) < 0$. The cost of computing these intermediate 
points is asymptotically the same as that of isolating the real roots of $A$, 
which is $\sOB(d^3 + d^2 \tau)$ \cite{pan-rootfinding-jsc}. 
The bitsize of such a $t$ is $\sOO(d \tau)$, e.g.,~\cite[Theorem~1]{emt-dmm-j}. 
This matches the bitsize of the root separation bound of $A$, 
since $t$ lies between two distinct real roots.

The evaluation of $A$ at $t$ can be performed in $\sOB(d^2 \tau)$ 
bit operations~\cite{bz-upol-eval-11,hn-upol-eval-11}.

\begin{lemma}[Witness point]
	Let $A$ be a univariate polynomial of size $(d, \tau)$. If $A$ is not nonnegative, 
	then there exists a rational number $t \in \QQ$ of bitsize $\sOO(d \tau)$ such that $A(t) < 0$.
    We can compute $t$ in $\sOB(d^3 + d^2 \tau)$ bit operations
    and  we can verify the inequality $A(t) < 0$ in $\sOB(d^2 \tau)$.
\end{lemma}

\subsection{Positivity over an interval {$[a, b]$}}
\label{sec:pos-over-interval}

We study certificates of positivity over an interval $[a, b]$, where $a, b \in \QQ_{\geq 0}$, for the polynomial 
\[A = \sum\nolimits_{k=0}^d a_k x^k = a_d \prod\nolimits_{i=1}^d (x - \alpha_i) \in \ZZ[x]. \]
Let $A$ has size $(d, \tau)$ and the bitsize of $a$ and $b$ be bounded by $\sigma$.
Following Chevillard et al.\ \cite[Sec.~5.2.5]{chml-usos-alg-11}, we consider the transformation, 
\[\phi: x \mapsto \frac{a + b y^2}{1 + y^2} ,\] 
that in turn induces the following transformation for $A$:
\[
\Aphi \coloneq (1+y^2)^d \phi(A) = (1+y^2)^d A(\tfrac{a + b y^2}{1 + y^2})  \in \ZZ[y] , 
\] 
Now, $\Aphi$ is nonnegative over $\RR$ if and only if $A$ is nonnegative over $[a, b]$. 
Thus, we can use the results of Sec.~\ref{sec:wu-sos-presentation} 
to certify that $\Aphi$ is positive over $\RR$,
instead of certifying directly that $A$ is positive over an interval.

Notice that if the bitsize of $a$ and $b$ is at most $\sigma$,
then the bitsize of $\Aphi \in \ZZ[y]$ is $\sOO(\tau + d \sigma)$.
Thus, if we straightforwardly apply the complexity bounds of the previous section, 
then we end up with a bit complexity bound of $\sOB( d^3 \sigma + d^2 \tau)$.
However, we can save a factor of $d$ from the term involving $\sigma$,
if we study the effect of $\phi$ on the separation bound, that is the minimum distance between the roots, of $A$. 

A close look in the complexity analysis of \usos reveals that 
the two important quantities are the value of the perturbation, $\eps \coloneq 2^{-\sfb}$ 
and the separation bound of the input polynomial $A$.
We study both them for the transformed polynomial $\Aphi$. 
First we consider the separation bound. 

\subsubsection{The separation bound of $\Aphi$}

If we apply $\phi$ to $A$ and we clear denominators, then 
the resulting polynomial is 
\[
A_{\phi}(y) = (1 + y^2)^d\cdot A(\phi(x)) = a_d\cdot \prod_{i=1}^d{\Bigl((b - \alpha_i) y^2  + (a - \alpha_i)\Bigl)} 
	\in \QQ[y],
\] 
and its roots, for $i \in [d]$, are
\[
	\zeta_i^{\pm} = \pm \sqrt{\frac{\alpha_i - a}{b- \alpha_i}} \enspace.
\] 
We need to (lower) bound the separation bound for $A_{\phi}$, that 
is the quantity 
\[\Delta_i(A_{\phi}) = \abs{\zeta_i^{\pm} - \zeta_j^{\pm}},\] 
where $\zeta_j^{\pm}$ is the closest root to $\zeta_i^{\pm}$.

We consider the polynomial 
\begin{align*}
	F(\delta) & = (\delta - \zeta_i^{+} - \zeta_j^{+})(\delta - \zeta_i^{+} - \zeta_j^{-})
            (\delta - \zeta_i^{-} - \zeta_j^{+})(\delta - \zeta_i^{-} - \zeta_j^{-}) \\
&  = \delta^{4} +
	2 \left(\frac{a- \alpha_{i}}{b -\alpha_{i}}+\frac{a-\alpha_{j}}{b -\alpha_{j}}\right) \delta^{2}
	+\frac{\left(\alpha_{i}-\alpha_{j}\right)^{2} \left(a -b \right)^{2}}{\left(b -\alpha_{i}\right)^{2} \left(b -\alpha_{j}\right)^{2}}
	\enspace. 
\end{align*}

Notice that among the roots of $F \in \CC[\delta]$ 
is the separation bound of $A_{\phi}$,
that is the difference $\zeta_i^{\pm} - \zeta_j^{\pm}$.
Hence, if we compute a lower bound for the roots of $F$, we also obtain a lower bound for $\Delta_i(A_{\phi})$.

Following \cite[Theorem~1]{emt-dmm-j}, we have that
\[
	\Delta(A_{\phi}) \geq 2 \frac{\tc(F)}{\normi{F}} \geq 
	\frac{\Delta_i(A)^2 \, (b-a)}{(b - \alpha_i)(b - \alpha_j) [ 2(ab + \alpha_i\alpha_j) - (\alpha_i + \alpha_j)(a +b)]}
	\geq 
	 2^{-6\sigma - 4\tau - 8} \Delta_i(A)^2 ,
\]
where we use the inequalities \[\abs{\alpha_i} \leq 2^{\tau+2},\quad
\abs{b - a} \geq 2^{-2\sigma}, \quad\text{and}\quad \abs{b - \alpha_i} \leq 2^{\sigma + \tau + 2}.\]
Consequently, 
\[\prod_{i} \Delta_i(A_{\phi}) \geq \prod_{i} \Delta_i(A) \, 2^{-6 \sigma - 4 \tau - 8} = 2^{-\sOO(d \sigma + d \tau)}.\]

This leads to the following lemma
\begin{lemma}
	\label{lem:Aphi-sep}
	Assume $A \in \ZZ[x]$ of size $(d, \tau)$
	and $a,b \in \QQ_{\geq 0}$ of bitsize $\sigma$.
	Then, 
	$-\lg\prod_i\Delta_i(\Aphi) = \sOO(d\sigma + d \tau)$.
\end{lemma}

\subsubsection{The perturbation $\eps$ for $\Aphi$}

The bound on $\eps \coloneq 2^{-\sfb}$ depends on a (lower) bound 
of the minimum of a polynomial $A$, in our case, on a minimum of $\Aphi$.
If $\xi_i$ are the roots of the derivative of $\Aphi$,
then $\eps$ depends on a lower bound on $\abs{\Aphi(\xi_i)}$.
To obtain this lower bound we will not rely on Lem.~\ref{lem:eval-A-at-da}
as we did in Sec.~\ref{sec:wpu-sos-steps-complexity}.
Instead we will rely on the following lemma from
this will save us a factor of $d$ in the complexity.

\begin{lemma}{\cite[Lemma~3]{pt-refine-jsc-16}}
  \label{lem:poly-eval-lower-bound} 
  Consider a square-free $A \in \RR[x]$ of degree $d$, and let its real roots be $\alpha_i$.
  Let  $x_0 \in \RR$ be such that $\abs{x_0 - \alpha_i} \geq \Delta_i/ c$
  for all {\em real} $\alpha_i$ such that $i \not=1$
  and $c \geq 2$.
  Then
  \begin{displaymath}
    \abs{A(x_0)} > \abs{\lc(A)} \,\abs{x_0 - \alpha_1} \, c^{1-d} \, \mathcal{M}(A)^{-1} 2^{\lg \prod_i \Delta_i(A) -1}
    \enspace,
  \end{displaymath}
  where $\mathcal{M}(A)$ is the Mahler measure of $A$.
\end{lemma}

Let $\xi \in \RR$ be the root of $\Aphi'$ where the minimum of $\Aphi$ is attained. 
Based on Dimitrov~\cite[Theorem~1]{Dimitrov-GL-98}
$\abs{\zeta_i - \xi} \geq \Delta_i(\Aphi) /d$, for all $i$.
Also $\mathcal{M}(\Aphi) \leq \normt{\Aphi}$.
In addition $\prod_i\Delta_i(\Aphi) \geq 2^{-\sOO(d\sigma + d \tau)}$
(Lem.~\ref{lem:Aphi-sep}).

Thus, 
\[ \abs{\Aphi(\beta)} \geq \abs{\lc(\Aphi)} \, \abs{\beta - \alpha_1} 
\, d^{1-d} \, \normt{\Aphi} \, 2^{\lg \prod_i \Delta_i(\Aphi) -1} ,
\]
which results in an $\eps = 2^{-\sfb}$,
where $\sfb = \sOO(d\sigma + d \tau)$.

\begin{lemma}[The value of $\eps$ for $\Aphi$]
	If $\eps = 2^{-\sfb}$, with $\sfb = \sOO(d\sigma + d \tau)$, then 
$\Aphie(x) = \Aphi(x) - \eps M(x) >0$ for all $x \in \RR$, and 
$-\lg\prod_i\Delta(\Aphi) = -\lg\prod_i\Delta(\Aphie) = \sOO(d\sigma + d \tau)$.
\end{lemma}

\subsubsection{Overall complexity}

\begin{theorem}[\usos at an interval]
	\label{thm:wpusos-interval}
	Let $A \in \ZZ[x]$ of size $(d, \tau)$.
	The algorithm \usos provides a certificate of positivity of $A$ 
	over an interval $[a, b]$, where $a$ and $b$ are rationals of bitsize $\sigma$,
	in $\sOB(d^3 + d^2 \sigma + d^2 \tau)$. 
\end{theorem}
\begin{proof}
	The computation of $A_{\phi}$ consists of a series of Mobius transformations. The most 
	computationally expensive is the (Taylor) shift, that is the transformation $x \mapsto x + a$, where $a$ is a rational of bitsize $\sigma$.
	This costs $\sOB(d^2 \sigma + d \tau)$ bit operations and results in a polynomial of bitsize $\sOO(\tau + d \sigma)$,
	\cite{vzgg-shift-97}.
	
	The input to \usos is $\Aphi$.		
	By choosing an $\eps \coloneq 2^{-\sfb}$, such that $\sfb = \sOO(d\sigma + d \tau)$,	we guarantee that $\Aphie = \Aphi - \eps M$ is positive over $\RR$ and 
	that the separation bounds of $\Aphi$ and $\Aphie$ are asymptotically the same; in our case $\sOO(d\sigma + d \tau)$.
The approximation of the roots of $\Aphie$ up to precision $\kappa$ costs $\sOB(d^3 + d^2\tau + d^2\tau + d \kappa)$,
	where $\kappa = \sOO(d\sigma + d \tau)$. 
	This cost dominates the overall complexity of the algorithm. 
\end{proof}

\begin{remark}
	If we target a certificate of positivity for the interval $(0, \infty)$, 
	then it suffices to consider the map $x \mapsto y^2$,
	and so there is no dependency on $\sigma$.
\end{remark}

To recover a weighted SOS representation of $A$ from a weighted SOS decomposition
of $\Aphi$, we proceed as follows \cite{chml-usos-alg-11}:
In the representation of $\Aphi$,
we decompose each polynomial $s_j$ into terms of odd and even degree, that is
\[
	\Aphi(y) = \sum\nolimits_{j=1}^{\nu} w_j \, s_j(y)^2 = 
	 \sum_j w_j\, s_{j,e}(y^2)^2 + w_j\, y^2\, s_{j,o}(y^2)^2.
\]
Then, we invert the change of variable using
$y^2 \mapsto \tfrac{x -a}{b -x}$ and $1 + y^2 \mapsto \tfrac{b-a}{b-x}$.
Consequently, after clearing denominators, we distinguish two cases.

\noindent
If $d$ is even, then 
\[
A(x) = \sum_j \tfrac{w_j}{(b - a)^d} \left[ (b - x)^{\frac{d}{2}} \,s_{j,e} \bigl( \tfrac{x - a}{b - x} \bigr) \right]^2 
+ (x - a)(b - x) \sum_j \tfrac{w_j}{(b - a)^d} \left[ (b - x)^{\frac{d}{2} - 1} \,s_{j,o} \bigl( \tfrac{x - a}{b - x} \bigr) \right]^2.
\]
If $d$ is odd, then 
\[
A(x) = (b - x) \sum_j \tfrac{w_j}{(b - a)^d} \left[ (b - x)^{\frac{d - 1}{2}} \,s_{j,e} \bigl( \tfrac{x - a}{b - x} \bigr) \right]^2 
+ (x - a) \sum_j \tfrac{w_j}{(b - a)^d} \left[ (b - x)^{\frac{d - 1}{2}} \,s_{j,o} \bigl( \tfrac{x - a}{b - x} \bigr) \right]^2.
\]

In both cases, it is important to notice that the bitsize of the polynomials in the representation is $\sOO(d\sigma + d \tau)$.

\section{Perturbed SOS certificate}
\label{sec:pert-sos}

Even though \usos algorithm is very efficient, it has the drawback that it forces us to compute (approximate) the roots of the polynomial $\Ae$,
instead of the input polynomial $A$. This has the consequence, 
that we should work with precision $\OO(d \tau)$ bits right from the beginning of the algorithm. 
Even more, the positivity certificate it corresponds to,
involves $\OO(d)$ sums of squares, while we know that it is theoretically possible to construct certificates with a smaller number of summands \cite{pourcet-5squares}. 
We try to leverage these weaknesses, 
by introducing an alternative, yet closely related, certificates that we call \emph{perturbed SOS certificate} for rational univariate polynomials. 
These, certify the nonnegativity for carefully chosen approximations (or perturbations) of the input polynomials. 
We demonstrate that, if the approximation is closed-enough (under a norm), then the nonnegativity of the perturbed polynomial is equivalent to the nonnegativity of the original polynomial. 
In this way, the perturbed SOS certificate involves only two squares of polynomials with rational coefficients.

\begin{theorem} \label{thm:approxImpliesNonNeg}
    Let $A \in \ZZ[x]$ be  of even degree $d$ and maximum coefficient bitsize $\tau$. If there is a nonnegative polynomial $B \in \QQ[x]$ such that $\normi{A-B} < 2^{-\sfb}$, with 
    $\sfb > 4 d \tau + 16d\lg{d}$,
    then $A$ is nonnegative.
\end{theorem}

\begin{proof}
For all the (real) roots, say $\alpha$, of $A$ it holds
$\abs{\alpha} < 2^{\tau+2}$, e.g.~\cite[Theorem~1]{emt-dmm-j}.
Moreover, as the leading term of $A$ is positive and $d$ is even, 
when $x$ goes to $\pm \infty$, then $A(x)$ goes to $+ \infty$. Hence, $A(x) > 0$ for every $x$ such that  $\abs{x} > 2^{\tau + 2}$.

Now let $\abs{x} \leq 2^{\tau +2}$, 
that is $x \in [-2^{\tau +2}, 2^{\tau +2}] =: J$. We will bound the maximum value of function $\abs{(A-B)(x)}$ 
in this interval. It holds 
\[
    \max_{x \in J} \abs{(A-B)(x)} 
    \leq \big(2^{(\tau + 2)}\big)^d (d+1) \, \normi{A-B} 
    \leq 2^{(\tau + 2)d + \lg{d} + 1 - \sfb} < 2^{-4 d \tau - 16d\lg{d}}.
\]
As $B(x)$ is a positive function, we have that for any $x \in J$,
    it holds 
    $0 \leq B(x) \leq A(x) + 2^{-4 d \tau - 16d\lg{d}}.$
As $A(x) > 0$  when $x \not\in J$,
    we deduce that 
    \begin{equation}
    	\label{eq:Aplus2b}	
	    0 \leq A(x) + 2^{-4 d \tau - 16d\lg{d}} ,
	    \qquad \text{for all } x \in \RR .
    \end{equation}
	To conclude, it remain to prove that the previous inequality, \eqref{eq:Aplus2b}, 
	implies the nonnegativity of $A(x)$. 
	We argue by contradiction.
Assume that $A(x)$ is not nonnegative. 
    As $A(x) \geq 0$, $\abs{x} > 2^{\tau + 2}$, there must have a negative critical value. Assume that $x^{\star}$ is a critical point realizing the biggest, strictly-negative, critical value.
By Lem.~\ref{lem:eval-A-at-da}, we have that $A(x^{\star}) < -2^{-4 d \tau - 16d\lg{d}}$, as $x^{\star}$ is a root of $A'$ but not root of $A$. However, this implies that $A(x^{\star}) + 2^{-4 d \tau - 16d\lg{d}} < 0$, which contradicts \eqref{eq:Aplus2b}.
\end{proof}

\begin{definition}
	\label{def:pert-SOS-cert}
	Let $A \in \QQ$ be of size $(d, \tau)$, where $d$ is even,
	that is nonnegative over $\RR$.
Fix rational polynomials $s_1,\dots, s_r \in \mathbb Q[x]$ and positive rational constants $\lambda_1,\dots,\lambda_r \in \QQ_{\geq 0}$, and let $B(x) := \sum_{i=1}^r \lambda_i s_i^2$. 
If $\normi{A-B} < 2^{-\sfb}$, with $\sfb > 4 d \tau + 16 d \lg d$,
    then $B$ is a \emph{perturbed SOS certificate} of nonnegativity of $A$.
\end{definition}

\begin{remark}
    Recall that, if $A(x)$ is square-free, then $A(x) \geq 0$ is equivalent to $A(x) > 0$. Moreover, we can certify that a polynomial is square-free by certifying that the greatest common divisor of $A(x)$ and $A'(x)$ is $1$ via B\'ezout's identity. Therefore, in the square-free case, the perturbed SOS certificates of nonnegativity lead to certificates of positivity. 
\end{remark}
 
 By combining all the previous ideas, we can exploit the SOS representation of 
 \eqref{eq:pos-certificate-R} to obtain a perturbed SOS certificate for a square-free polynomial $A$. 
 For this, we will approximate the roots of $A$ up to a precision 
 that will guarantee that the induced polynomials $\wt{P}$ and $\wt{Q}$
 form a perturbed SOS certificate of positivity for $A$. 
The following theorem gives the details of this approach.

\begin{theorem}
	\label{thm:pert-sos-cert-bits}
    Let $A \in \ZZ[x]$ be a positive square-free of even degree $d = 2m$ and maximum coefficient bitsize $\tau$.
Let $\alpha_j^{+} = \gamma_j + \imath \, \delta_j$, where
    $\gamma_j \in \RR$, $\delta_j \in \RR_{\geq 0}$, for $j \in [m]$;
    that is $\alpha_j^{+}$ are the roots with positive imaginary part.
    
    Let  $\wt{\gamma}_j \in \QQ$, $\wt{\delta}_j \in \QQ_{> 0}$, for $j \in [m]$,
    be rational approximations up to precision $2^{-\lambda}$, for $\lambda = 9 d \tau + 60 d \lg d$. That is
    \[
    	\abs{\alpha_i - (\wt{\gamma}_j - \imath  \wt{\delta}_j)} \leq 2^{-\lambda}.
    \]
Then, the polynomial $\wt{P}^2 + \wt{Q}^2$, where $\wt{P}, \wt{Q} \in \QQ[x]$ are  defined as 
    \[ 
    	\wt{P}(x) + \imath \, \wt{Q}(x) = 
    	\prod\nolimits_{j = 1}^{m} (x - \wt{\gamma}_j - \imath \, \wt{\delta}_j).
    \]
    is a perturbed SOS certificate of positivity of $A$.
\end{theorem}

\begin{proof}
	All the roots of $A$ have magnitude smaller that $2^{\tau + 2}$, e.g.~\cite[Theorem~1]{emt-dmm-j}.
	Using this, and the fact that $\lambda > 9 d \tau + 60 \, d \lg d$, 
	Lem.~\ref{lem:fan-in}, implies the following inequality
    \[
   \Normi{A(x) - 
   	\lc(A) \prod\nolimits_{j = 1}^{m}(x - \wt{\gamma}_j - \imath \, \wt{\delta}_j)
   		                    (x - \wt{\gamma}_j + \imath \, \wt{\delta}_j)} 
    < 2^{-\lambda + (4d (\tau + 2) + 32\,d - 7) + 1} 
    < 2^{- 5 d \tau - 19 d \lg d}.
    \]
    Moreover, following the arguments in Step~2 of Sec.~\ref{sec:wu-sos-presentation}
   	\[
    \wt{P}(x)^2 + \wt{Q}(x)^2 = 
    \lc(A) \prod\nolimits_{j = 1}^{m}(x - \wt{\gamma}_j - \imath \, \wt{\delta}_j)
   		                    (x - \wt{\gamma}_j + \imath \, \wt{\delta}) .
    \]
    Therefore, from Def.~\ref{def:pert-SOS-cert} and Thm.~\ref{thm:approxImpliesNonNeg},
    $\wt{P}^2 + \wt{Q}^2$ is a perturbed SOS certificate of nonnegativity for $A(x)$. As $A(x)$ is square-free, we conclude that $A(x)$ is positive.
\end{proof}

Regarding the complexity of computing the perturbed certificate,
we have the following theorem.
\begin{theorem}
	\label{thm:pert-sos-cert-complexity}
	Let $A \in \ZZ[x]$ be a square-free positive polynomial of over $\RR$, of size $(d, \tau)$.
	We can compute a perturbed SOS certificate of positivity of $A$
	in $\sOB(d^3 + d^2 \tau)$.
	The certificate consists of the polynomial $\wt{P}^2 + \wt{Q}^2$
	that has bitsize  $\sOO(d \tau)$. 
\end{theorem}

\begin{proof}
The computation of a perturbed SOS certificate for $A$, 
requires us to (efficiently) approximate the  roots $A(x)$ (up to any desired precision), say $\lambda$.
In our case, $\lambda = \sOO(d \tau)$ (Thm.~\ref{thm:pert-sos-cert-bits}).
The approximation of all the roots requires $\sOB(d^3 + d^2 \tau)$ bit operations (Thm.~\ref{thm:root-isol-approx}).
The computation of $\wt{P}$ and $\wt{Q}$ is based on Lem.~\ref{lem:fan-in}
and also costs $\sOB(d^3 + d^2 \tau)$. 
Lem.~\ref{lem:fan-in} shows that the bitsizes of $\wt{P}$ and $\wt{Q}$ are $\sOO(d \tau)$.
\end{proof}

\section{Interlacing polynomials and T-systems}
\label{sec:interlacing}

For a given square-free polynomial $A \in \RR[x]$, positive over $\RR$, 
$(0, \infty)$, or any interval $[a, b] \subset \RR$, 
we establish a connection, actually an equivalence, 
between an SOS representation of $A$, 
and the T-systems and Karlin points~\cite{Dio-T-systes-arxiv}. 
In this way, we make the theory of T-systems constructive
we show that the Karlin points are the real roots of certain 
interlacing polynomials.
We refer the reader to \Cref{sec:T-systems} 
for the (very) basic definitions and to \cite{Dio-T-systes-arxiv} for further details 
on T-systems and positive polynomials.

\subsection{Strict positivity and Karlin points over $\RR$}
\label{sec:Karlin-R}

Consider the square-free polynomial $A(x) =  \sum_{k=0}^d a_k x^k \in \RR[x]$,
of even degree $d$ and $a_d >0$.
We further assume that $A$ is strictly positive over $\RR$, hence it does not have any real roots; its factorization to linear factors is
\begin{equation}
	\label{eq:A-root-factor}
	A(x) = a_d \prod_{i=1}^{d}(x - \alpha_i) = 
	a_d  \prod_{i=1}^{d/2} (x - \gamma_i + \imath\, \delta_i)
	\prod_{i=1}^{d/2} (x - \gamma_i - \imath\, \delta_i),
\end{equation}
where we assume that $\delta_i > 0$ for all $i \in [d/2]$.
We define the polynomials $P, Q \in \RR[x]$ in such a way that the following equations hold,
\begin{equation}
	\label{eq:P-pm-Q}
	\prod_{i=1}^{d/2} (x - \gamma_i + \imath\, \delta_i) = 
	P(x) + \imath Q(x)
	\quad \text{ and } \quad
	\prod_{i=1}^{d/2} (x - \gamma_i - \imath\, \delta_i) = 
	P(x) - \imath Q(x) .
\end{equation}

\begin{lemma}[Positivity over $\RR$ and interlacing]
	\label{lem:PQ-interlace}
	If $A  =  \sum_{k=0}^d a_k x^k \in \RR[x]$, of degree $d = 2m$ and $a_d > 0$, 
	is strictly positive over $\RR$, let $P, Q \in \RR[x]$ be as defined above. Then, we have that
	$A(x) = a_d \, P(x)^2 + a_d \, Q(x)^2$ and the polynomials $P$ and $Q$
	are interlacing of degrees $m$ and $m-1$, respectively.
\end{lemma}
\begin{proof}
	\Cref{eq:A-root-factor,eq:P-pm-Q} imply that $A$ is a sum of two squares.
If there are no assumptions on $\delta_i$'s, 
then there are $2^{d/2 -1}$ possible distinct pairs
of $\Set{P(x) \pm \imath Q(x)}$, e.g.,~\cite{reznick-concrete17-00},
and consequently, $2^{d/2 -1}$ inequivalent representations of $A$ as a sum of two squares.
However, by the uniqueness of the decomposition in \Cref{cor:tsystemsAllLine}, the only representation that results on $P$ and $Q$ being interlacing are the ones that we obtain when  we impose the condition $\delta_i > 0$, or equivalently $\delta_i < 0$,
for all $i \in [d/2]$. 
This is no coincidence, as the claim is a direct consequence of the Hermite-Biehler (HB) theorem
(Theorem~\ref{thm:HB}), 
that states that  $P(x)  + \imath Q(x)$ has not roots in the upper half plane
if and only if $P$ and $Q$ interlace. 
\end{proof}

In the Appendix, we present an alternative proof of Lemma~\ref{lem:PQ-interlace}, 
which also could be seen as a proof of HB theorem.
It is not based on analytic properties of polynomials,
as is the proof(s) of HB, e.g.~\cite[Thm.~6.3.4]{rs-athp-2002},
but relies directly on the factorization of a univariate polynomial 
to linear factors. 
It seems, from our point of view at least, that our proof is conceptually simpler,
and potentially of independent interest. 
For a restatement of the HB theorem 
in the modern language of real-stable polynomials and proper position 
we refer the reader to the work of Br\"anden, e.g.~\cite{branden-matroid-07}.

\begin{remark}
	An interesting remark is that since $P$ and $Q$ are real rooted, 
	their number of positive real roots is \emph{exactly} the number
	of sign variations in their coefficients list. This follows from
	Descartes' rule of sign.
\end{remark}

From the point of view of T-systems, \Cref{sec:T-systems} and \Cref{thm:Karlin}, 
for a strictly positive polynomial $A$ over $\RR$
the polynomials 
$P^2$ and $Q^2$ of Lamma~\ref{lem:PQ-interlace} 
play the role of $f^{*}$ and $f_{*}$.
In our case $\cF = \{1, x, x^2, \dots, x^d\}$.
Also, the Karlin points, see \Cref{eq:K-xi-yj},
are unique and interlacing. The latter is also a property of the real roots of $P$ and $Q$. Hence, we have the following corollary
that characterizes the Karlin points and allows us to compute them.

\begin{corollary}[Karlin points over $\RR$]
	\label{cor:Karlin-pts-R}
	Let $A \in \RR[x]$ of degree $d = 2m$ be strictly positive over $\RR$, 
	and $A(x) = P(x)^2 + Q(x)^2$, where $P, Q \in \RR[x]$, 
	as in Lemma~\ref{lem:PQ-interlace}.
	The real roots of $P$ and $Q$ are the Karlin points.
\end{corollary}

\subsection{Strict positivity and Karlin points over $[0, \infty)$}
\label{sec:Karlin-interval}

Consider $A(x) =  \sum_{k=0}^d a_k x^k \in \RR[x]$,
where $a_d >0$ and $a_0 \not= 0$.
We assume that $A$ is strictly positive over $[0, \infty)$;
so, it does not have nonnegative real roots.
In this case, Pol\'ya and Szeg\"o proved, the now classical result,
that $A$ has the representation $A(x) = f(x) + x\, g(x)$, where $f$ and $g$ 
are sum of squares of real polynomials
and $\deg(f), \deg(x\, g) \leq \deg(A)$, 
e.g.,~\cite[Prop.~2]{pr-univ-pos-int-00}.

We will show that if $A$ is strictly positive over $[0, \infty)$, then 
$A$ has a representation $A = \rP^2 + x\,  \rQ^2$, where $\rP$ and $\rQ$ 
are interlacing polynomials
and their real roots are the Karlin points of $A$ in $[0, \infty)$.
We use the notation $\rP$ and $\rQ$ to highlight that these
polynomials are different from the polynomials $P$ and $Q$ of Lemma~\ref{lem:PQ-interlace}.

\begin{theorem}[Strict positivity over $[0, \infty)$]
	\label{lem:PQ-interlace-0-inf}
	Let $A =  \sum_{k=0}^d a_k x^k \in \RR[x]$, of degree $d$ and $a_d > 0$, 
	be strictly positive over $[0, \infty)$.
	Then, $A(x) = a_d \, \rP(x)^2 + a_d\, x\, \rQ(x)^2$, where the polynomials $\rP, \rQ \in \RR[x]$ are interlacing and have only positive real roots.	
	If $A$ has even degree, $\deg(A) = d = 2m$, 
	then $\deg(P) = \tfrac{d}{2}$ and $\deg(Q) = \tfrac{d}{2} - 1$.
If $A$ has odd degree, $\deg(A) = d = 2m + 1$, 
	then $\deg(P) = \tfrac{d-1}{2} = \deg(Q)$.
\end{theorem} 
	
To simplify the calculations, we assume that $A$ is monic, that is $a_d = 1$.

\begin{proof}
Assume that $A(x)$ is positive in the interval $(0,\infty)$. 
Then, $A(x)$ admits the following factorization
\[
A(x) = \prod_{i=1}^r (x + \gamma_i) \prod_{j=1}^t (x - \alpha_j - \imath \beta_j)(x - \alpha_j + \imath \beta_j),
\]
    where $d = r + 2t$, $\gamma_i, \beta_i > 0$ and $\alpha_i \in \mathbb R$.

    Consider now the polynomial $A(x^2)$. The factorization of this polynomial is 
    \[A(x^2) = \prod_{i=1}^r (x - \imath \sqrt{\gamma_i} )(x + \imath \sqrt{\gamma_i} ) 
    \prod_{j=1}^t 
    (x - \eta^+_j + \imath \eta^-_j )
    (x + \eta^+_j - \imath \eta^-_j ) 
    (x - \eta^+_j - \imath \eta^-_j )
    (x + \eta^+_j + \imath \eta^-_j ) \enspace, 
    \]
    where
    $\eta^{\pm}_j = \sqrt{\frac{\sqrt{\alpha_j^2+\beta_j^2} \pm \alpha_j}{2}}$, $\gamma_i,\beta_i > 0$ and $\alpha_i \in \mathbb R$.

Notice that $A(x^2)$ is strictly positive in $\mathbb R$, so by Lemma~\ref{lem:PQ-interlace}, there are two real polynomials $P$ and $Q$ with only reals interleaving roots, such that $A(x^2) = P^2 + Q^2$ and,
\begin{align*}
P + \imath\, Q  &= 
\prod_{i=1}^r (x - \imath\, \sqrt{\gamma_i})
    \prod_{j=1}^t 
    (x + \eta^+_j - \imath\, \eta^-_j ) 
    (x - \eta^+_j - \imath\, \eta^-_j )
    \\
    &= (\imath)^r \prod_{i=1}^r (\sqrt{\gamma_i} - \imath\, x)
    \prod_{j=1}^t 
    (x^2 - (\eta^+_j)^2 - (\eta^-_j)^2 - \imath\,  2 \eta^-_j x ) \enspace.
    \end{align*}
We introduce a new variable $\omega$ and we set $\omega = \imath\, x$. 
Then, there are real constants $\{c_{a,b}\}_{a,b}$ such that,
\begin{align*}
    P + \imath\, Q &= 
    (\imath)^r \prod_{i=1}^r (\sqrt{\gamma_i} - \omega)
    \prod_{j=1}^t 
    (x^2 - (\eta^+_j)^2 - (\eta^-_j)^2 - 2 \eta^-_j  \omega)   
    = 
    (\imath)^r \sum_{a,b} c_{a,b} (x^2)^a \omega^b
    \\
    &= (\imath)^r \Big(\sum_{a,k} c_{a,2k} (x^2)^a \omega^{2k}
    +
    \sum_{a,k} c_{a,2k+1} (x^2)^a \omega^{2k+1} \Big) \enspace.
\end{align*}
By combining all the identities, we get that
\[
P + \imath\, Q =
     (\imath)^r \Big(\sum_{a,k} (-1)^{k} c_{a,2k} (x^2)^{a+k} 
    +
    \imath\,  x \sum_{a,k} (-1)^{k} c_{a,2k+1} (x^2)^{a+k} \Big)
=:
(\imath)^r (\rP(x^2) + \imath\,  x \rQ(x^2)) .
\]
Therefore, it holds 
\[
A(x^2) = P^2 + Q^2 = (\imath)^r (-\imath)^r (\rP(x^2) + \imath\,  x \rQ(x^2)) (\rP(x^2) - \imath\, x \rQ(x^2)) = 
\rP(x^2)^2 + x^2 \rQ(x^2)^2 ,
\]
and so $A(x) = \rP(x)^2 + x\, \rQ(x)^2$. To see that the roots of $\rP$ and $\rQ$ interlace, with no loss of generality, assume $r$ is even. Then,
$P(x) = (-1)^{r/2} \rP(x^2)$ and $Q(x) = (-1)^{r/2} x \rQ(x^2)$.
By Lemma~\ref{lem:PQ-interlace}, $P$ and $Q$ are real rooted and interlacing.
Hence, $\rP$ and $\rQ$ have only real positive roots. 
Also they are interlacing, as the square root of every root of $\rP$, respectively $\rQ$, is a root of $P$, respectively $Q$.
\end{proof}

As in the case of $\RR$, the next corollary follows straightforwardly 
if we combine Lemma~\ref{lem:PQ-interlace-0-inf} with the uniqueness property
of Karlin points. 

\begin{corollary}[Karlin points over $[0, \infty)$]
	Let $A \in \RR[x]$, of degree $d$, be positive over $(0, \infty)$, 
	and $A(x) = \rP(x)^2 + x \rQ(x)^2$, where $\rP, \rQ \in \RR[x]$, 
	as in Lemma~\ref{lem:PQ-interlace-0-inf}.
	The real roots of $\rP$ and $\rQ$ are the Karlin points of $A$ in $(0, \infty)$.
\end{corollary}

\medskip

What about strict positivity over an arbitrary interval $I \subseteq \RR$?
To study this case, we follow closely Powers and Reznick \cite{pr-univ-pos-int-00}.
If a polynomial $A$ is strictly positive over an interval $[a, b] \subseteq \RR$, 
then by the transformation
\[
B(x) = A\Big( \tfrac{(b - a)x + (b + a)}{2} \Big),
\]
we obtain a polynomial $B$ that is strictly positive over $[-1,1]$.
The transformation involves homothecy and translation; both of them preserve interlacing. 
Hence, it suffices to consider the strict positivity of $A$ over $[-1,1]$.
Even more, strict positivity over $[0, \infty)$ 
and strict positivity over $[-1, 1]$ are closely related through the Goursat transform. 
The $d$-th degree Goursat transform is 
\[
\cG[A] (x)= (1 + x)^d \, A \Big( \tfrac{1 - x}{1 + x} \Big).
\]
If we apply $\cG$ two times, then we notice that it is almost, up to a constant depending on the degree, its own inverse, 
that is 
\[
\cG[\cG[A]](x) = (1 + x)^d \,\cG[A]\Big( \tfrac{1 - x}{1 + x} \Big)
= (1 + x)^d \Big(1 + \tfrac{1 - x}{1 + x}\Big)^d
A \Big( \tfrac{1 - \frac{1 - x}{1 + x}}{1 + \frac{1 - x}{1 + x}} \Big)
= 2^d A(x).
\]
 
 So, Goursat's lemma states that 
  $A$ is strictly positive over $[-1,1]$ 
 if and only if $\cG[A]$ is strictly positive over $[0, \infty)$ and $\deg(\cG[A]) = d$
 \cite[Lemma~1]{pr-univ-pos-int-00}.
 Similarly, $A$ is strictly positive over $[0, \infty)$ if and only if 
 $\cG[A]$ is strictly positive over $[-1, 1]$ and $\deg(\cG[A]) \leq d$.
 
 The application of $\cG$ consists of a composition of homothecies, translations, and inversions.
 All three preserve interlacing, hence it suffices to study the positivity,
 and compute the Karlin points,  
  on $(0, \infty)$.
  
  Even though from a theoretical point of view $(0, \infty)$  
  and $[-1,1]$ are almost equivalent, with respect to positivity certificates, this is not, exactly the case, 
  from a practical point of view. Goursat's transform might increase
  the norm of the polynomial and hence we might be forced to compute with bigger coefficients. We do not exploit this direction further.

 \section{Implementation and experiments}
 \label{sec:implementation}
 
 We provide an open source implementation of \usos in \textsc{maple}, 
 available in \href{https://gitlab.inria.fr/cert/upos.git}{UPOS}.
 We also perform some  experiments to demonstrate the efficiency
 of the algorithm and our implementation.
 The software is easy to use and, besides \textsc{maple}, 
 does not rely on external software packages.

\subsection{Some insights on the implementation}
\label{sec:implementation-details}

Regarding the actual implementation of \usos,  
there are certain choices that play an important role in practice.

\noindent
\emph{Computing the perturbation $\eps$.}
The first and most important decision concerns how to compute the perturbation
$\eps$ in Step~1 of the algorithm.
An obvious choice would be to use the a-priori worst case bound
for the minimum of the input polynomial $A$,
see Lem.~\ref{lem:eval-A-at-da}. This is a bad choice,
as  we are forcing the implementation to \emph{always} operate 
with the worst case number of bits
and slows down the running times dramatically. We do not consider this case at all.
Our goal should be to compute with the actual minimum of $A$
and not rely on worst case bounds; which very often overestimate a lot the bitsize. 

In this context, a straightforward strategy to compute $eps$  consists in starting with an initial value, say $\frac{1}{2}$, and repeatedly dividing by 2,
until the polynomial $\Ae := A - \eps M $ has no real roots
\emph{and} positive leading coefficient. Hence, each iteration performs a call
to a real root-finding algorithm.
This is the strategy implemented in the \cite{mss-wusos-alg-19}.
Unfortunately, this approach also has significant theoretical and practical drawbacks.
The worst case bound for $\eps$ is, roughly, $2^{-\sOO(d \tau)}$
(Cor.~\ref{cor:Ae-roots-approx-complexity}, Lem.~\ref{lem:eps-value}),
thus, with this approach, 
we need to perform $\sOO(d \tau)$ iterations, in the worst case, 
to compute $\eps$. This leads to a bit complexity
bound $\sOB(d^3 \tau^2)$, which is not anymore linear in the bitsize of the input! 
From a practical point of view, this choice has the consequence that it forces the 
number of iteration to depend on the bitsize of $\eps$, which, if it is big, slows down significantly the running times,
even when the worst case bitsize bounds are not attained.

A natural question to ask is if we can come up with an efficient implementation to compute $\eps$ that, at the same time, does not dominate or alter the overall worst case bit complexity of \usos. 

To overcome the previous, theoretical and practical, obstacles
we rely on the important and simply observation that we can consider the dyadic representation of the perturbation, $\eps = 2^{-\sfb}$; notice that $\sfb = \sOO(d \tau)$, roughly (Lem.~\ref{lem:eps-value}).
Then, instead computing $\eps$, our goal should be to compute $\sfb$. 
We do so by performing exponential binary search starting from zero.
To find $2^{-\sfb}$, we start at exponent $e=0$, i.e., $2^0=1$, and
successively double it: $1,2,4,\dots$ until $e$ exceeds $\sfb$.  This exponential
phase yields an interval $(e/2,\, e]$ that must contain $\sfb$.  Then,
we perform binary search on this interval to compute the exact exponent
$\sfb$.
The exponential phase doubles the exponent until it exceeds $\sfb$, requiring
$\OO(\log \sfb)$ steps.  The subsequent binary search on the interval
$(e/2,\, e]$ also takes $\OO(\lg \sfb)$ steps.  Thus, the overall complexity of
exponential binary search on the exponent to find $\eps := 2^{-\sfb}$ is
$\OO(\log \sfb) = \sOO( \lg(d \tau))$ iterations.
Recall that corresponds to number of times we call the real root-finding algorithm
to test whether that polynomial $A - 2^{-e} M$ has real roots or not. 
This approach has the advantages of all the other approaches
and none of their drawbacks.
It does not alter the worst case complexity bound
and from a practical point of view the gain is also enormous.

In Table~\ref{tab:exp-eps} we present a comparison on various approaches for 
computing $\eps$, on some positive polynomials coming from \cite{chml-usos-alg-11};
see also the next section for further details.
The first three columns are: the id of the polynomial, its degree and its input bitsize.
The fourth columns is the computed bound on $\eps := 2^{\sfb}$; actually its exponent.
The last three columns present the number of steps that various algorithms
perform to compute $\eps$.
The column \texttt{SubDiv} is the subdivision method, starting from $\tfrac{1}{2}$
and repeatedly dividing by 2. The column  \texttt{Min+ExpBin}
corresponds to a hybrid algorithm that first approximates the minimum the polynomial
and then finds the best dyadic bound using exponential search.
Finally, the last column,  \texttt{ExpBin}, is the (pure) exponential binary search algorithm that we described earlier. Evidently exponential binary search is the most efficient 
algorithm for computing $\eps$.
For a demonstration of how the computation of $\eps$ affects 
the actual running time of \usos, we refer to Table~\ref{tab:w-comparison}.
In this set of experiments, we consider 
the three algorithms on modified Wilkinson polynomials, see Eq.~\ref{eq:modif-Wilk-poly} and also next section. Again, the exponential binary search approach
is best approach. 
Our implementation of \usos offers all three approaches for computing $\eps$.

\begin{table}[h]
\centering
\begin{tabular}{|c|r|r||r||r|r|r|}
\hline
\# & Degree&  \makecell{Input \\bitsize} &  $\eps = 2^{-\sfb}$ & 
	\makecell{\texttt{SubDiv} \\ \#(iters)} & \makecell{\texttt{Min+ExpBin} \\ \#(iters)} & \makecell{ \texttt{ExpBin} \\ \#(iters)} \\
\hline
 1  & 13 & 359  &  60 &  61 & 15 & 13 \\
 3  & 32 & 439  & 116 & 117 & 27 & 15 \\
 4  & 22 & 492  & 103 & 104 & 24 & 15 \\
 5  & 34 & 775  & 235 & 236 & 27 & 17 \\
 6  & 17 & 190  &  69 &  70 & 30 & 15 \\
 7  & 43 & 371  &  70 &  71 & 27 & 15 \\
 8  & 22 & 275  &  74 &  75 & 24 & 15 \\
 9  & 20 & 353  &  45 &  46 & 27 & 13 \\
10  & 25 & 312  &  67 &  68 & 24 & 15 \\
\hline
\end{tabular}
\caption{Comparison of the number of steps that various approaches perform to compute a a bound on $\eps$; the initial perturbation of \usos. The polynomials are from \cite{chml-usos-alg-11}.}
\label{tab:exp-eps}
\end{table}

\medskip

\noindent
\emph{Deciding the working precision.}
Another important choice concerns the (initial) precision 
we use to approximate the roots of the polynomial $\Ae$.
The implementation of (complex) root finding should avoid
working right from the beginning with the worst case theoretical bounds.
It does not seem very likely that we will need to compute with
(that) many bits of precision. We expect that the separation bound,
that is the minimum distance between the roots of a polynomial
not to be very small. Indeed this is the case for (a wide variety of) ``random'' polynomials \cite{etct-bwc-issac-22}.
Following this discussion, our implementation starts
with some initial precision and if the corresponding inequalities 
are not satisfied, Eq.~\eqref{eq:eps-ineq}, then we double it. 
Our experiments suggest that a good practical heuristic is to consider 
a starting accuracy that depends on the bitsize of the perturbation $\eps$.
To approximate the complex roots of a polynomial we use the
build-in \texttt{Isolate} function of \maple.
It realizes the algorithm of Imbach and Moroz \cite{im-slv-23}.
To refine roots, up to any desired precision, 
we use the build-in \maple function \texttt{hefroots:-refine} which, unfortunately, to the date, is without documentation.

\medskip
\noindent
\emph{A disclaimer.}
We should mention, that even though it is easier to provide mathematical software on top of well established computer algebra packages, like \maple, 
this convenience comes with certain limitations.
For example, we cannot work exclusively with numbers that are powers of two;
these are the only type of numbers that we need for the perturbation $\eps$
and the approximation of the roots of $\Ae$.
Such implementation tricks can speed the actual running times by several orders of magnitude, e.g.~\cite{te-jcf-98,rz-desc-04}.
In \maple, it is very difficult, if possible at all, to implement 
such tricks, while it is not that complicated in standard programming 
languages, like  C/C++.

\subsection{Experiments with positive polynomials}
\label{sec:experiments}

We perform various experiments to demonstrate the efficiency of our implementation
and to study its practical behavior.
All the experiments were performed using \textsc{maple 25}, on MacBook Air with an Apple M2 cpu, having 24GB of memory, and running Sonoma 14.7.5.
The running times presented, at the last column of every table, are the average of 10 runs.

The first set of experiments that we performed are on polynomials coming from \cite{chml-usos-alg-11}. This is a set of nine polynomials that we need to certify that they are positive in a (small) interval having rational coefficients. 
The results appear in Table~\ref{tab:sollya}. The first column is the index of the polynomial, the second its degree, the third its bitsize, the fourth is the maximum bitsize of the polynomials and rationals in the SOS decomposition, and the last one is the time needed by our implementation of \usos to compute the rational SOS decomposition. Even though our implementation is quite efficient and the bitsize of the output is reasonable, it is difficult to draw general conclusion
as the polynomials are varying difficulty, i.e., different bitsizes, minimum, separations bounds, etc.

\begin{table}[ht]
	\centering
        \begin{tabular}{|r||r|r|r|r|}
		\hline
		\makecell{\#} & \makecell{Degree} & \makecell{Input \\bitsize} & \makecell{Output \\bitsize} & \makecell{Time\\(ms)} \\
  		\hline
 1 &	 13 &	        359 &	       2\,655 &	   7 \\ 
 3 &	 32 &	        439 &	       8\,084 &	   62 \\ 
 4 &	 22 &	        492 &	       4\,351 &	   28 \\ 
 5 &	 34 &	        775 &	      16\,675 &	  171 \\ 
 6 &	 17 &	        190 &	       3\,227 &	   15 \\ 
 7 &	 43 &	        371 &	      14\,540 &	  176 \\ 
 8 &	 22 &	        275 &	       4\,040 &	   26 \\ 
 9 &	 20 &	        353 &	       2\,620 &	    8 \\ 
10 &	 25 &	        312 &	       3\,226 &	    6 \\ 
\hline		
\end{tabular}
    	\captionof{table}{Positive polynomials in an interval from \cite{chml-usos-alg-11}.}
    	\label{tab:sollya}
\end{table}

Another set of experiments considers modified Wilkinson
polynomials,  inspired from \cite{mss-wusos-alg-19}.
These are polynomials of the form 
\begin{equation}
	\label{eq:modif-Wilk-poly}
 	A(X) = \prod_{i=2}^n(x - i)^2 - x^2/11237 + 1.
\end{equation}
The output data of this set appears in Table~\ref{tab:w}.
Again we observe a linear fit concerning 
the theoretical and the experimental bound on the bitsize of the SOS decomposition. 
The graph and the linear equation appears in Fig.~\ref{fig:w-fit}.
We notice a fluctuation in the output bitsize, which is probably due to the fact
that these polynomials are not good representatives of ``generic'' positive polynomials.
Nevertheless, our implementation computes a weighted SOS representation efficiently.

\begin{table}[ht]
    \centering
    \begin{minipage}{0.3\textwidth}
        \centering
\begin{tabular}{|r|r|r|r|}
	\hline
	\makecell{Degree} & \makecell{Input \\bitsize} & \makecell{Output \\bitsize} & \makecell{Time\\(ms)} \\
  \hline
 10 & 17  & 721        & 31      \\
12 & 22  & 870        & 27      \\
14 & 28  & 1\,283     & 49      \\
16 & 34  & 2\,472     & 53      \\
18 & 41  & 2\,568     & 77      \\
20 & 48  & 7\,527     & 146     \\
22 & 55  & 8\,133     & 196     \\
24 & 62  & 8\,557     & 265     \\
26 & 70  & 9\,118     & 145     \\
28 & 78  & 8\,675     & 162     \\
30 & 86  & 8\,447     & 302     \\
32 & 94  & 24\,294    & 822     \\
34 & 102 & 25\,567    & 1\,080  \\
36 & 111 & 11\,077    & 1\,052  \\
38 & 119 & 28\,004    & 1\,536  \\
40 & 128 & 28\,733    & 1\,826  \\
    \hline
\end{tabular}
	\captionof{table}{Modified Wilkinson polynomials}
	\label{tab:w}
    \end{minipage}
    \hfill
    \begin{minipage}{0.6\textwidth}
        \centering
        \includegraphics[scale=0.38]{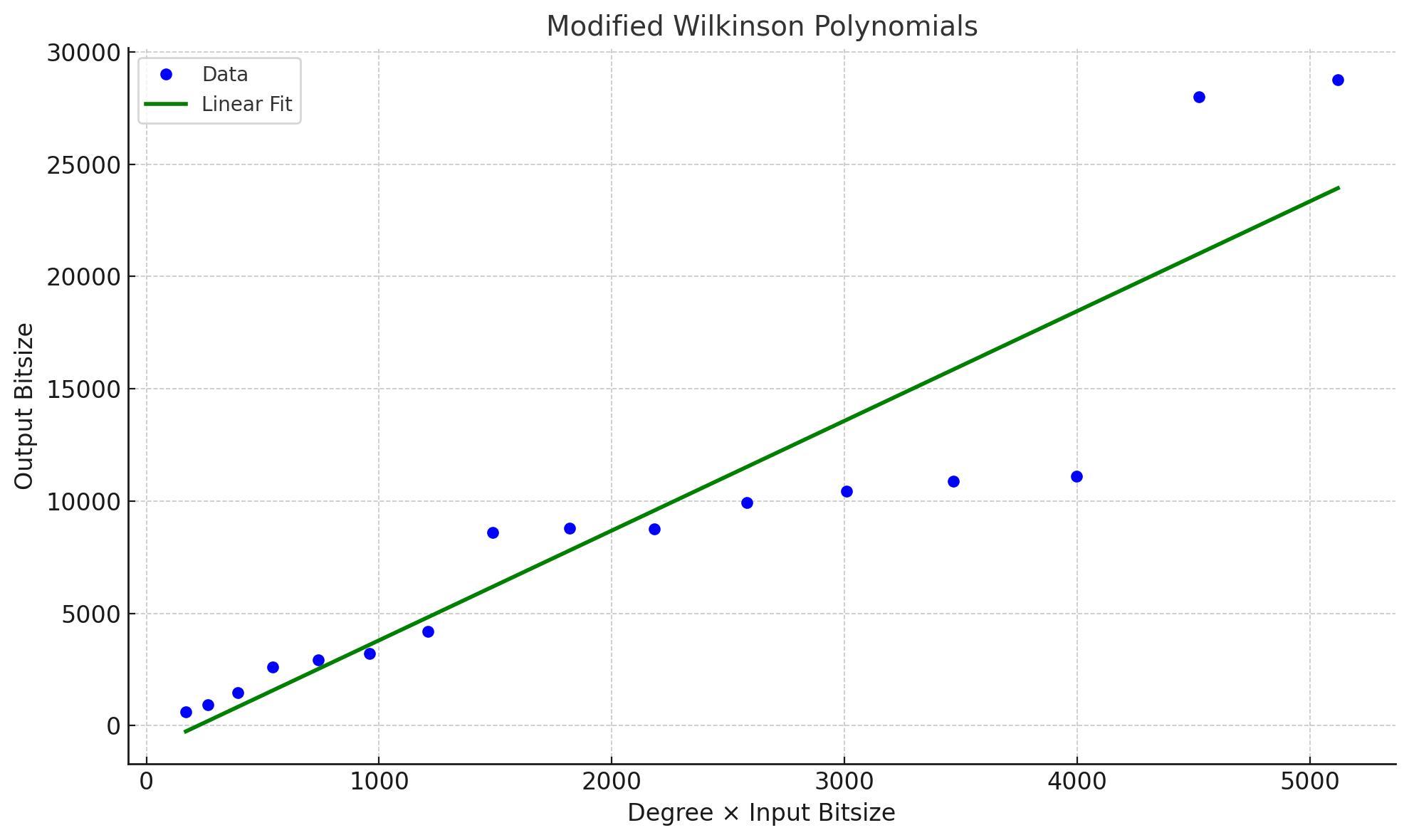}
        \captionof{figure}{The linear fit of positive Wilkinson polys is \\$\text{Output Bitsize} \simeq 4.8866 (d\cdot  \text{ Input Bitsize}) - 1085.69$}
        \label{fig:w-fit}
    \end{minipage}
\end{table}

\begin{table}[ht]
    \centering
\begin{tabular}{|r|r|r||r||r||r|}
	\hline
	\makecell{Degree} & \makecell{Input \\bitsize} & \makecell{Output \\bitsize} & \makecell{\texttt{SubDiv} \\ Time (ms) } & \makecell{\texttt{Min+ExpBin} \\ Time (ms)} & \makecell{\texttt{ExpBin} \\ Time (ms)} \\
  \hline
10 &         17 &           721  &      40    &      36    &      21 \\
 12 &         22 &          870  &      45    &      31    &      24 \\
 14 &         28 &       1\,283  &     101    &      54    &      47 \\
 16 &         34 &       2\,472  &     116    &      66    &      49 \\
 18 &         41 &       2\,568  &     179    &     139    &      75 \\
 20 &         48 &       7\,527  &     362    &     214    &     145 \\
 22 &         55 &       8\,133  &     635    &     323    &     207 \\
 24 &         62 &       8\,557  &   1\,075   &     450    &     255 \\
 26 &         70 &       9\,118  &     728    &     159    &     143 \\
 28 &         78 &       8\,675  &     879    &     164    &     152 \\
 30 &         86 &       8\,447  &   1\,316   &     525    &     291 \\
 32 &         94 &      24\,294  &   2\,425   &   1\,309   &     770 \\
 34 &        102 &      25\,567  &   4\,328   &   1\,690   &     983 \\
 36 &        111 &      11\,077  &   6\,295   &   1\,833   &     994 \\
 38 &        119 &      28\,004  &   9\,777   &   2\,837   &   1\,470 \\
 40 &        128 &      28\,733  &  12\,717   &   3\,396   &   1\,840 \\
    \hline
\end{tabular}
	\captionof{table}{Running times of \usos to compute the SOS representation of modified Wilkinson polynomials, Eq.~\eqref{eq:modif-Wilk-poly}, based on different methods to compute the initial perturbation $\eps$.}
	\label{tab:w-comparison}
\end{table}

\begin{table}[ht]
    \centering
    \begin{minipage}{0.3\textwidth}
        \centering
        \begin{tabular}{|r|r|r|r|}
		\hline
		\makecell{Deg} & \makecell{Input \\bsz} & \makecell{Output \\bsz} & \makecell{Time\\(ms)} \\
  		\hline
 20 &	         82 &	       1\,161 &	    2 \\ 
 40 &	         81 &	       2\,248 &	    2 \\ 
 60 &	         83 &	       3\,296 &	    4 \\ 
 80 &	         83 &	       4\,379 &	    5 \\ 
 100 &	         83 &	       5\,440 &	    6 \\ 
 120 &	         83 &	       6\,520 &	    7 \\ 
 140 &	         83 &	       7\,577 &	   15 \\ 
 160 &	         84 &	       8\,671 &	   17 \\ 
 180 &	         84 &	       9\,754 &	   39 \\ 
 200 &	         84 &	      10\,777 &	   49 \\ 
 220 &	         84 &	      11\,871 &	   91 \\ 
 240 &	         84 &	      12\,980 &	  248 \\ 
 260 &	         84 &	      14\,001 &	 1\,715 \\ 
   	\hline
		\end{tabular}
		\captionof{table}{Sum of 3 squares of random polynomials.}
		\label{tab:sos-3}
    \end{minipage}
    \hfill
    \begin{minipage}{0.3\textwidth}
       \centering
        \begin{tabular}{|r|r|r|r|}
		\hline
		\makecell{Deg} & \makecell{Input \\bsz} & \makecell{Output \\bsz} & \makecell{Time\\(ms)} \\
  		\hline
  20 &	         83 &	       1\,139 &	    2 \\ 
 40 &	         84 &	       2\,222 &	    1 \\ 
 60 &	         84 &	       3\,292 &	    2 \\ 
 80 &	         84 &	       4\,346 &	    3 \\ 
 100 &	         84 &	       5\,426 &	    5 \\ 
 120 &	         84 &	       6\,520 &	    5 \\ 
 140 &	         85 &	       7\,581 &	    6 \\ 
 160 &	         85 &	       8\,631 &	   12 \\ 
 180 &	         85 &	       9\,718 &	   17 \\ 
 200 &	         84 &	      10\,830 &	   22 \\ 
 220 &	         85 &	      11\,848 &	   45 \\ 
 240 &	         85 &	      12\,905 &	   91 \\ 
 260 &	         85 &	      14\,029 &	  369 \\ 
    	\hline
		\end{tabular}
    	\captionof{table}{Sum of 11 squares of random polynomials.}       
    	\label{tab:sos-11}
    \end{minipage}
        \hfill
    \begin{minipage}{0.3\textwidth}
       \centering
        \begin{tabular}{|r|r|r|r|}
		\hline
		\makecell{Deg} & \makecell{Input \\bsz} & \makecell{Output \\bsz} & \makecell{Time\\(ms)} \\
  		\hline
 20 &	         84 &	       1\,143 &	    1 \\ 
 40 &	         84 &	       2\,214 &	    1 \\ 
 60 &	         85 &	       3\,281 &	    3 \\ 
 80 &	         85 &	       4\,349 &	    3 \\ 
 100 &	         85 &	       5\,406 &	    4 \\ 
 120 &	         85 &	       6\,532 &	    4 \\ 
 140 &	         85 &	       7\,595 &	    6 \\ 
 160 &	         85 &	       8\,662 &	    9 \\ 
 180 &	         85 &	       9\,704 &	   16 \\ 
 200 &	         86 &	      10\,789 &	   26 \\ 
 220 &	         86 &	      11\,841 &	   54 \\ 
 240 &	         85 &	      12\,935 &	   88 \\ 
 260 &	         86 &	      14\,019 &	  483 \\ 
    	\hline
		\end{tabular}
		\captionof{table}{Sum of 31 squares of random polynomials.}       
		\label{tab:sos-31}
    \end{minipage}
\end{table}

\begin{figure}[ht]
    \centering
    \begin{minipage}{0.45\textwidth}
   \centering
        \includegraphics[scale=0.31]{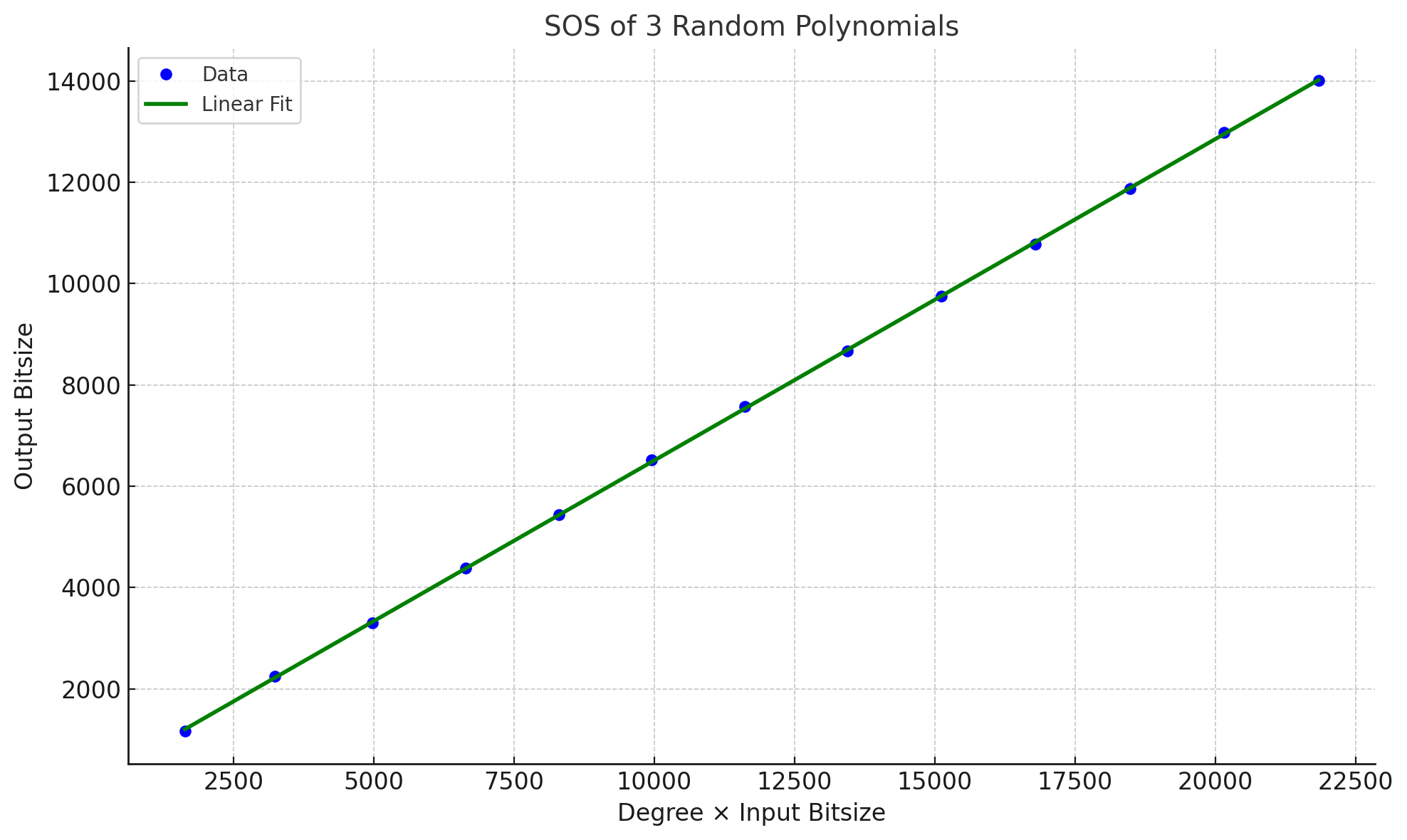}
        \caption{The linear fit of 3 squares is \\
        $\text{Output Bsz} \simeq 0.6342 (d \cdot \text{ Input Bsz}) + 163.91$}
        \label{fig:lin-fit-3} 
    \end{minipage}
    \hfill
    \begin{minipage}{0.45\textwidth}
        \centering
        \includegraphics[scale=0.31]{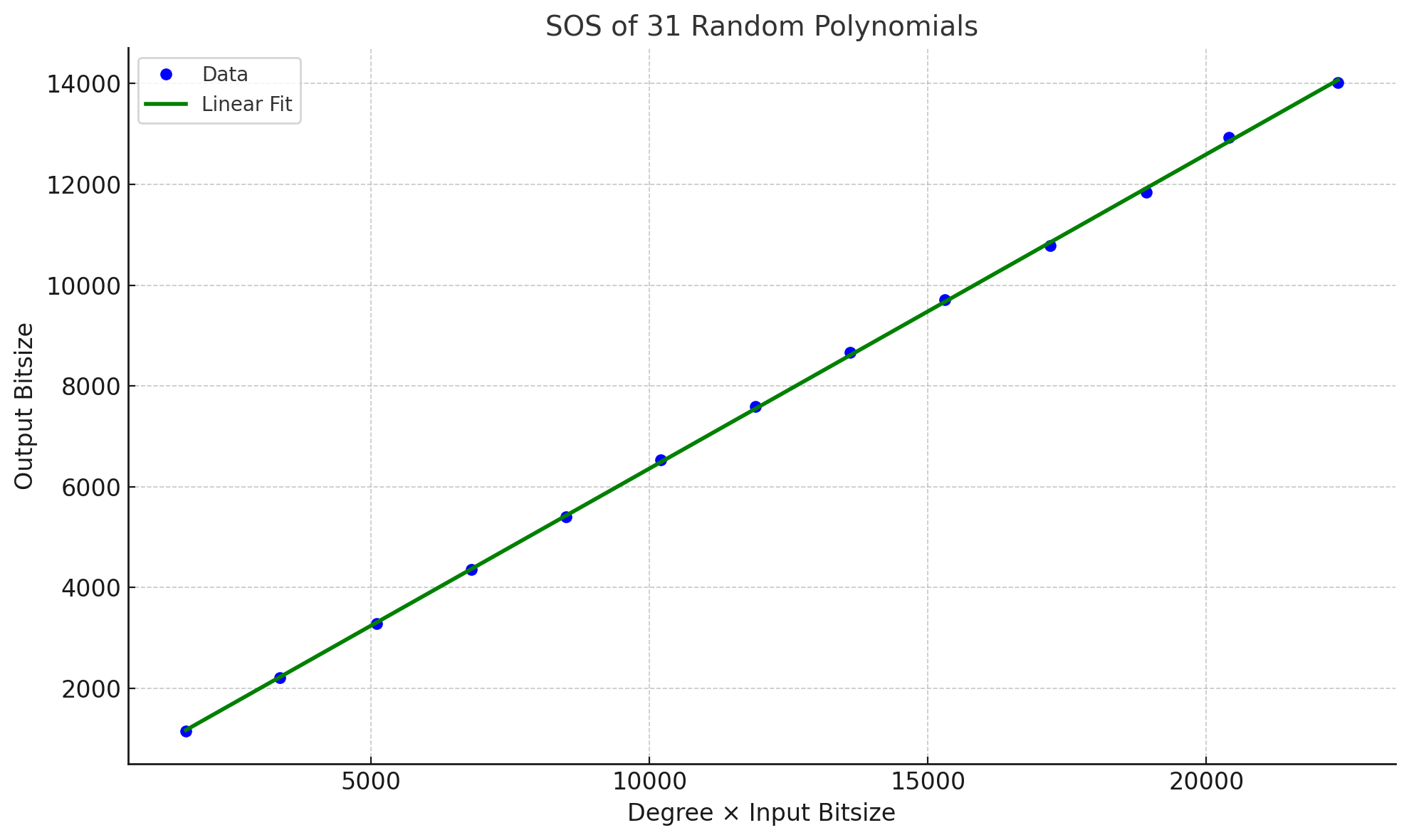}
        \caption{The linear fit of 31 squares is \\
        $\text{Output Bsz} \simeq 0.6239 (d \cdot \text{ Input Bsz}) + 120.82$}
        \label{fig:lin-fit-31}
    \end{minipage}
\end{figure}

It is rather a difficult task to consider random 
positive univariate polynomials. Such polynomials would be useful
to study the practical behaviour of \usos.
Unfortunately, this task requires us to sample uniformly from the 
convex cone of positive polynomials. This is 
computationally very expensive, even if we consider
a polyhedral approximation of the corresponding cone \cite{ergur-approx-pos-19}.
A rather good compromise\footnote{We thank Alperen Erg\"ur for this useful suggestion.}
is to consider random polynomials 
and the sum of their squares. 
Hence, the second set of experiments considers polynomials of the form 
$A(x) = \sum_{i=1}^{\nu}A_i(x)^2$,
where we sample the integer coefficients of polynomials $A_i$
uniformly at random from the interval $[-2^{40}, 2^{40}]$;
thus, the bitsize of $A$ is around 80.
Tables~\ref{tab:sos-3}, \ref{tab:sos-11}, and \ref{tab:sos-31} 
present the results of the experiments for $\nu \in \Set{3,11,31}$ SOS summands.
Figures~\ref{fig:lin-fit-3} and \ref{fig:lin-fit-31} present the graphs of bitsize of the polynomials in the decomposition vs the product of the degree of $A$ and its bitsize, for the two extreme cases $3$ and $31$.
Recall, by \Cref{lem:bsz-of-certif}, 
the bitsize of the decomposition is $\sOO(d \tau)$.
We notice that there is a precise linear fit, as predicted by the theory. 
In the figures we also mention the explicit linear equations.

\section{Conclusion}
\label{sec:conclusion}

We have presented improved complexity bounds for computing rational weighted sums of squares (SOS) certificates for univariate polynomials positive over $\RR$ or rational intervals, refining previous analyses and reducing the bit complexity by a factor of~$d$. Beyond the algorithmic advances, we uncovered new structural properties: the SOS summands form an interlacing pair, revealing a connection to Karlin points and T-systems. These insights contribute to our understanding of positivity certificates and their geometric interpretation. 
	Our open-source \maple implementation confirms the practical efficiency of the algorithm. 
	Future directions include considering sparse certificates for sparse univariate polynomials.

\section*{Acknowledgements}

We thank Mario Kummer for pointing out that Lemma~\ref{lem:PQ-interlace} corresponds to Hermite-Biehler theorem.
The authors thank John May and J\"urgen Gerhard for their help with the undocumented Maple function \texttt{hefroots:-refine}. 
This research benefited from the support of the FMJH Program Gaspard Monge for optimization and operations research and their interactions with data science. MB and ET are partially supported 
by the PGMO grant SOAP, 
ANR JCJC PeACE (ANR-25-CE48-3760), 
and ANR PRC ZADyG (ANR-25-CE48-7058).
PJdD is supported by the Deutsche Forschungs\-gemein\-schaft DFG with the grant DI-2780/2-1 and his research fellowship at the Zukunfts\-kolleg of the University of Konstanz, funded as part of the Excellence Strategy of the German Federal and State Government.

\bibliographystyle{abbrv}

\begin{thebibliography}{10}

\bibitem{bkm-eff-pos-24}
L.~Baldi, T.~Krick, and B.~Mourrain.
\newblock An effective positivstellensatz over the rational numbers for finite
  semialgebraic sets.
\newblock {\em arXiv preprint arXiv:2410.04845}, 2024.

\bibitem{bz-upol-eval-11}
M.~Bodrato and A.~Zanoni.
\newblock Long integers and polynomial evaluation with estrin's scheme.
\newblock In {\em 2011 13th International Symposium on Symbolic and Numeric
  Algorithms for Scientific Computing}, pages 39--46. IEEE, 2011.

\bibitem{BorMoe-jcss-74}
A.~Borodin and R.~Moenck.
\newblock Fast modular transforms.
\newblock {\em Journal of Computer and System Sciences}, 8(3):366–386, June
  1974.

\bibitem{bcr-bern-pos-dcg}
F.~Boudaoud, F.~Caruso, and M.-F. Roy.
\newblock Certificates of positivity in the bernstein basis.
\newblock {\em Discrete \& Computational Geometry}, 39(4):639--655, 2008.

\bibitem{branden-matroid-07}
P.~Br{\"a}nd{\'e}n.
\newblock Polynomials with the half-plane property and matroid theory.
\newblock {\em Advances in Mathematics}, 216(1):302--320, 2007.

\bibitem{chml-usos-alg-11}
S.~Chevillard, J.~Harrison, M.~Jolde{\c{s}}, and C.~Lauter.
\newblock Efficient and accurate computation of upper bounds of approximation
  errors.
\newblock {\em Theoretical Computer Science}, 412(16):1523--1543, 2011.

\bibitem{papp-dual-cert-bsz}
M.~M. Davis and D.~Papp.
\newblock Rational dual certificates for weighted sums-of-squares polynomials
  with boundable bit size.
\newblock {\em Journal of Symbolic Computation}, 121:102254, Mar. 2024.

\bibitem{Dio-T-systes-arxiv}
P.~J. di~Dio.
\newblock {An introduction to T-Systems-with a special emphasis on sparse
  moment problems, sparse Positivstellens\"atze, and sparse
  Nichtnegativstellens\"atze}.
\newblock {\em arXiv preprint arXiv:2403.04548}, 2024.

\bibitem{Dimitrov-GL-98}
D.~Dimitrov.
\newblock A refinement of the gauss-lucas theorem.
\newblock {\em Proceedings of the American Mathematical Society},
  126(7):2065--2070, 1998.

\bibitem{dressler2017positivstellensatz}
M.~Dressler, S.~Iliman, and T.~De~Wolff.
\newblock A positivstellensatz for sums of nonnegative circuit polynomials.
\newblock {\em SIAM Journal on Applied Algebra and Geometry}, 1(1):536--555,
  2017.

\bibitem{emt-dmm-j}
I.~Emiris, B.~Mourrain, and E.~Tsigaridas.
\newblock Separation bounds for polynomial systems.
\newblock {\em Journal of Symbolic Computation}, 101:128--151, 2020.

\bibitem{etct-bwc-issac-22}
A.~Erg{\"u}r, J.~Tonelli-Cueto, and E.~Tsigaridas.
\newblock Beyond worst-case analysis for root isolation algorithms.
\newblock In {\em Proc Intl. Symposium on Symbolic and Algebraic Computation
  (ISSAC)}, pages 139--148, 2022.

\bibitem{ergur-approx-pos-19}
A.~A. Erg\"ur.
\newblock Approximating nonnegative polynomials via spectral sparsification.
\newblock {\em SIAM Journal on Optimization}, 29(1):852--873, 2019.

\bibitem{Fisk-interlace-bk}
S.~Fisk.
\newblock Polynomials, roots, and interlacing.
\newblock {\em arXiv preprint math/0612833}, 2006.

\bibitem{hn-upol-eval-11}
W.~Hart and A.~Novocin.
\newblock Practical divide-and-conquer algorithms for polynomial arithmetic.
\newblock In {\em International Workshop on Computer Algebra in Scientific
  Computing}, pages 200--214. Springer, 2011.

\bibitem{im-slv-23}
R.~Imbach and G.~Moroz.
\newblock Fast evaluation and root finding for polynomials with floating-point
  coefficients.
\newblock In {\em Proc Intl. Symposium on Symbolic and Algebraic Computation
  (ISSAC)}, pages 325--334, 2023.

\bibitem{Karlin-repr-pos-63}
S.~Karlin.
\newblock Representation theorems for positive functions.
\newblock {\em Journal of Mathematics and Mechanics}, 12(4):599--617, 1963.

\bibitem{KarStu-Tsystems-66}
S.~Karlin and W.~J. Studden.
\newblock {\em Tchebycheff systems: With applications in analysis and
  statistics}.
\newblock Interscience Publishers, New York, London, Sydney, 1966.

\bibitem{kmv-prouchet-23}
P.~Koprowski, V.~Magron, and T.~Vaccon.
\newblock Pourchet’s theorem in action: decomposing univariate nonnegative
  polynomials as sums of five squares.
\newblock In {\em Proc 2023 International Symposium on Symbolic and Algebraic
  Computation (ISSAC)}, pages 425--433, 2023.

\bibitem{kms-usos-22}
T.~Krick, B.~Mourrain, and A.~Szanto.
\newblock Univariate rational sums of squares.
\newblock {\em Revista de la Uni{\'o}n Matem{\'a}tica Argentina},
  64(2):215--237, 2022.

\bibitem{Landau-usos-06}
E.~Landau.
\newblock {\"U}ber die darstellung definiter funktionen durch quadrate.
\newblock {\em Mathematische Annalen}, 62(2):272--285, 1906.

\bibitem{mss-wusos-alg-19}
V.~Magron, M.~Safey El~Din, and M.~Schweighofer.
\newblock Algorithms for weighted sum of squares decomposition of non-negative
  univariate polynomials.
\newblock {\em Journal of Symbolic Computation}, 93:200--220, 2019.

\bibitem{msw-aprox-fact-15}
K.~Mehlhorn, M.~Sagraloff, and P.~Wang.
\newblock From approximate factorization to root isolation with application to
  cylindrical algebraic decomposition.
\newblock {\em Journal of Symbolic Computation}, 66:34--69, 2015.

\bibitem{pan-rootfinding-jsc}
V.~Y. Pan.
\newblock Univariate polynomials: Nearly optimal algorithms for numerical
  factorization and root-finding.
\newblock {\em Journal of Symbolic Computation}, 33(5):701–733, May 2002.

\bibitem{pt-struct-mat-j}
V.~Y. Pan and E.~Tsigaridas.
\newblock Nearly optimal computations with structured matrices.
\newblock {\em Theoretical Computer Science}, 681:117--137, 2017.

\bibitem{pt-refine-jsc-16}
V.~Y. Pan and E.~P. Tsigaridas.
\newblock Nearly optimal refinement of real roots of a univariate polynomial.
\newblock {\em Journal of Symbolic Computation}, 74:181–204, May 2016.

\bibitem{pourcet-5squares}
Y.~Pourchet.
\newblock Sur la repr{\'e}sentation en somme de carr{\'e}s des polyn{\^o}mes
  {\`a} une ind{\'e}termin{\'e}e sur un corps de nombres alg{\'e}briques.
\newblock {\em Acta Arithmetica}, 19(1):89--104, 1971.

\bibitem{powers-certif-book}
V.~Powers.
\newblock {\em Certificates of Positivity for Real Polynomials}.
\newblock Springer, 2021.

\bibitem{pr-univ-pos-int-00}
V.~Powers and B.~Reznick.
\newblock Polynomials that are positive on an interval.
\newblock {\em Transactions of the American Mathematical Society},
  352(10):4677--4692, 2000.

\bibitem{rs-athp-2002}
Q.~I. Rahman and G.~Schmeisser.
\newblock {\em Analytic theory of polynomials}.
\newblock Number~26. Oxford University Press, 2002.

\bibitem{reznick-concrete17-00}
B.~Reznick.
\newblock Some concrete aspects of hilbert's 17th problem.
\newblock {\em Contemporary mathematics}, 253(251-272), 2000.

\bibitem{rz-desc-04}
F.~Rouillier and P.~Zimmermann.
\newblock Efficient isolation of polynomial's real roots.
\newblock {\em Journal of Computational and Applied Mathematics},
  162(1):33--50, 2004.

\bibitem{Schw-sos-99}
M.~Schweighofer.
\newblock {Algorithmische Beweise f\"ur Nichtnegativ- und
  Positivstellens\"atze}.
\newblock Diplomarbeit, Universit\"at Passau, 1999.

\bibitem{sombra-mix-res-height}
M.~Sombra.
\newblock The height of the mixed sparse resultant.
\newblock {\em American Journal of Mathematics}, 126(6):1253--1260, 2004.

\bibitem{te-jcf-98}
E.~Tsigaridas and I.~Emiris.
\newblock On the complexity of real root isolation using continued fractions.
\newblock {\em Theoretical Computer Science}, 392(1-3):158--173, 2008.

\bibitem{vzgg-shift-97}
J.~Von Zur~Gathen and J.~Gerhard.
\newblock Fast algorithms for taylor shifts and certain difference equations.
\newblock In {\em Proc. International Symposium on Symbolic and Algebraic
  Computation (ISSAC)}, pages 40--47, 1997.

\end{thebibliography}

\newpage

\appendix

\section{Alternative proof for interlacing}
\label{sec:alternative-interlacing}

We present an alternative proof of Lemma~\ref{lem:PQ-interlace},
which is also an alternative proof of Hermite-Bielher theorem (Thm.~\ref{thm:HB})
that is based on induction and  direct manipulation of the roots of the polynomials.
The HB theorem is as follows:
\begin{theorem}[Hermite--Biehler~{\cite[Thm.~6.3.4]{rs-athp-2002}}]
	\label{thm:HB}
Let $F(x) = P(x) + i\,Q(x) \in \CC[x]$, where $P, Q \in \mathbb{R}[x]$.  
The following statements are equivalent:
\begin{enumerate}[(i)]
  \item All zeros of $F$ lie in the open upper half-plane
  $\Set{\zeta \in \mathbb{C} : \Im(\zeta) > 0 }$.
\item   
  The real polynomials $P$ and $Q$:
  	(a) are real rooted, 
  	(b) have roots that \emph{interlace}, and (c) 
  	have a Wronskian
      $W(x) := P'(x)\,Q(x) - P(x)\,Q'(x)$
      of  constant nonzero sign on $\RR$.
  \end{enumerate}
\end{theorem}

\begin{otherlemma*}[Positivity over $\RR$ and interlacing]
	\label{lem:other-PQ-interlace}
	If $A  =  \sum_{k=0}^d a_k x^k \in \RR[x]$, of degree $d = 2m$ and $a_d > 0$, 
	is strictly positive over $\RR$, let $P, Q \in \RR[x]$ be as defined above. Then, we have that
	$A(x) = a_d \, P(x)^2 + a_d \, Q(x)^2$ and the polynomials $P$ and $Q$
	are interlacing of degrees $m$ and $m-1$, respectively.
\end{otherlemma*}
\begin{proof}
    First observe that the identity $A(x) = a_d \, P(x)^2 + a_d \, Q(x)^2$ follows from the fact that, by definition of $P(x)$ and $Q(x)$ , $A(x) = a_d (P(x) + \imath\, Q(x)) (P(x) - \imath\, Q(x))$.

    Hence, we prove by induction on $m$ that $P(x)$ and $Q(x)$ have interlacing roots.
Notice that $m$ also corresponds to the number of products required for $P$ and $Q$.
	To make this explicit,	we write $P_m$ and $Q_m$. 
	It holds $\deg(P_m) = m$ and $\deg(Q_m) = m-1$.

	\medskip
	Induction start ($m = 2$, i.e., $d = 4$): 
	Following \eqref{eq:P-pm-Q}, 
	\[ (x - \gamma_1 - \imath\, \delta_1)(x - \gamma_2 - \imath\, \delta_2)  = 
	P_2(x) - \imath Q_2(x) ,
	\]
	where 
\[P_2(x) = x^{2} - (\gamma_{1}+\gamma_{2}) x + \gamma_{1} \gamma_{2} -\delta_{1} \delta_{2}
\quad\text{and}\quad Q_2(x) = -(\delta_{1}+\delta_{2}) x + \gamma_{1} \delta_{2} + \gamma_{2} \delta_{1}.\]

Regarding $Q_2$, it is of degree 1 and
has one real root
\[\xi := \frac{\gamma_{1} \delta_{2} + \gamma_{2} \delta_{1}}{\delta_{1}+\delta_{2}} \in \RR.\]
As for $P_2$, it has degree 2 and its discriminant is \[\mathtt{disc}(P_2) = (\gamma_1 - \gamma_2)^2 + 4 \delta_1 \delta_2 \geq 0.\] 
Hence, $P_2$ has 2 real roots, $\zeta_{\pm} \in \RR$,
\[
	\zeta_{-} := \frac{\gamma_1 + \gamma_2}{2} - \frac{\sqrt{(\gamma_1 - \gamma_2)^2 + 4 \delta_1 \delta_2}}{2}
	<
	\frac{\gamma_1 + \gamma_2}{2} + \frac{\sqrt{(\gamma_1 - \gamma_2)^2 + 4 \delta_1 \delta_2}}{2} =: \zeta_{+} \enspace.
\]
It is straightforward that $\zeta_{-} \leq \zeta_{+}$.
Overall, both $P_2$ and $Q_2$ are real rooted.

It remains to show that the roots of $P_2$ and $Q_2$ interlace,
that is, $\zeta_{-} < \xi < \zeta_{+}$.
Regarding the right inequality, we have 
\begin{align*}
	 \xi < \zeta_{+} && \Leftrightarrow &&
	 \frac{\gamma_{1} \delta_{2} + \gamma_{2} \delta_{1}}{\delta_{1}+\delta_{2}} 
	 <
	 \frac{\gamma_1 + \gamma_2}{2} + \frac{\sqrt{(\gamma_1 - \gamma_2)^2 + 4 \delta_1 \delta_2}}{2} \\
&& \Leftrightarrow &&
	 0< \frac{(\delta_{1}+\delta_{2})(\gamma_1 + \gamma_2) + (\delta_{1}+\delta_{2})\sqrt{(\gamma_1 - \gamma_2)^2 + 4 \delta_1 \delta_2} - 2(\gamma_{1} \delta_{2} + \gamma_{2} \delta_{1})  }{2(\delta_{1}+\delta_{2})} \\
&& \Leftrightarrow &&
	 0 < (\delta_{1}+\delta_{2})\sqrt{(\gamma_1 - \gamma_2)^2 + 4 \delta_1 \delta_2} +
	 (\gamma_1 - \gamma_2)(\delta_1 - \delta_2) \\
&& \Leftrightarrow &&
	 (\gamma_2 - \gamma_1) (\delta_1 - \delta_2) <
	 (\delta_{1}+\delta_{2}) \sqrt{(\gamma_1 - \gamma_2)^2 + 4 \delta_1 \delta_2)} \\
	 &&&& \text{\emph{(if the lhs is negative, then the inequality holds, so we assume it is positive)}}\\
&& \Leftrightarrow &&
	 (\gamma_2 - \gamma_1)^2 (\delta_1 - \delta_2)^2 <
	 (\delta_{1}+\delta_{2})^2 ((\gamma_1 - \gamma_2)^2 + 4 \delta_1 \delta_2) \\
	 &&&& \text{ \emph{(the lhs is positive, so we square both sides)}}\\
&& \Leftrightarrow &&
	 0< - (\gamma_2 - \gamma_1)^2 (\delta_1 - \delta_2)^2 +
	 (\delta_{1}+\delta_{2})^2 ((\gamma_1 - \gamma_2)^2 + 4 \delta_1 \delta_2) \\
&& \Leftrightarrow &&
	 0 < 4 \delta_1 \delta_2 ( (\delta_{1}+\delta_{2})^2 +  (\gamma_1 - \gamma_2)^2)
\end{align*}
Similarly we prove that $\zeta_{-} < \xi$. 

\medskip
\noindent
Induction step ($m\to m+1$):
Assume that the result is true for some $m$. That is, $P_m$ and $Q_m$
are interlacing and their degrees are $m$ and $m-1$, respectively.
We obtain the polynomials $P_{m+1}$ and $Q_{m+1}$
as follows
\begin{align*}
P_{m+1} - \imath Q_{m+1} & = \prod_{i=1}^{m+1} (x - \gamma_i - \imath\, \delta_i) 
= (x - \gamma_{m+1} - \imath\, \delta_{m+1}) \prod_{i=1}^{m} (x - \gamma_i - \imath\, \delta_i) \\
 & = (x - \gamma_{m+1} - \imath\, \delta_{m+1})( P_{m} - \imath Q_{m}) \\
 & = \Big((x - \gamma_{m+1}) P_m - \delta_{m+1} Q_m \Big)  
     - \imath \Big(\delta_{m+1} P_m + (x - \gamma_{m+1}) Q_m \Big) \enspace .
\end{align*}

If we write the computation of $P_{m+1}$ and $Q_{m+1}$
from $P_m$ and $Q_m$ in matrix form, then we have 
\[
\begin{pmatrix}
x - \gamma_{m+1} & -\delta_{m+1} \\
\delta_{m+1} & x - \gamma_{m+1}
\end{pmatrix}
\begin{pmatrix}
P_{m} \\
Q_{m}
\end{pmatrix} =
\begin{pmatrix}
P_{m+1} \\
Q_{m+1}
\end{pmatrix} .
\]
Following Fisk \cite[Cor.~3.54(4)]{Fisk-interlace-bk}, as $\delta_{m+1} > 0$,
the matrix multiplying the vector ${P_m \choose Q_m}$
preserves interlacing. 
Hence, if $P_{m}$ and $Q_{m}$ are interlacing, then 
so are the polynomials $P_{m+1}$ and $Q_{m+1}$.
By inspecting the computations, we deduce that the
degrees of resulting polynomials are $m+1$ and $m$, respectively.
This concludes the proof.
\end{proof}

\section{Useful algorithms and complexity bounds}
\label{sec:additional-alg-bounds}

We present some known results that are useful in our analysis.
In particular, we use results on the separation bounds 
of polynomials that are ``close'' with respect to the one norm,
algorithms for multiplication of polynomials,
for approximating their (complex) roots to any desired accuracy,  
and bounds on the minimum of a univariate polynomial.

\subsection{Preliminaries on root separation and approximation}
\label{sec:prelim-root-approx}

We exploit the work of Mehlhorn, Sagraloff, and Wang \cite{msw-aprox-fact-15} on root approximation and refinement 
for univariate polynomials. 
The lemmata that we present are simplified variants of the original
ones, as we assume we are working with square-free polynomials. 
We refer the reader to \cite{msw-aprox-fact-15} 
for the general versions and further details.
The root isolation algorithm assumes that there is an oracle that is able to provide 
rational approximations of the coefficients of the input polynomial up to arbitrary precision.

Consider the square free polynomial 
\[p(x) = \sum_{i=0}^{n}{p_i x^i} =  p_n \prod_{i=1}^n(x - z_i) \in \CC[x], \] 
where $z_i \in \CC$ are its roots.
Also let \[\hp(x) = p_n \prod_{i=1}^n(x - \hz_i).\]
We need the following notations:
\begin{itemize}
\item $M(x):=\max\{1, \abs{x}\}$, for $x\in \RR$,
\item $\tau_{p}$ is the minimal nonnegative integer such that
  $\frac{|p_{i}|}{|p_n|}\le2^{\tau_{p}}$ for all $i=0,\ldots,n-1$,   
\item $\Gamma_p:=M(\max_i (\log M(z_i)))$ denotes the \emph{logarithmic root bound} of $p$, and 
\item $\Delta_i = \min_{j \not= i} \abs{z_i - z_j}$ is the local separation bound.
\end{itemize}

The following lemma relates the separation bound(s)
of the roots of $p$ with the separation bound(s) of the roots of $\hp$,
when the polynomials are "sufficiently" close.

\begin{lemma}[{\cite[Lemma~3]{msw-aprox-fact-15}}]
	\label{lem:approx-sep}
	Consider $p = \sum_{i=0}^{n}{p_i x^i} \in \CC[x]$. 
	Also, let $\hp \in \CC[x]$ be such that 
	\[\normo{p-\hp} \le 2^{-\sfb} \normo{p}.\] 
	If, for all $i \in [n]$, 
\begin{align}
	\label{eq:sfb-cond-1}
	\sfb &\ge  \max\{8n,n \log ( n)\} \text{, and $\sfb$ is a power of two}, \\
\label{eq:sfb-cond-2}
	2^{-\sfb/2} & \le \frac{\Delta_i}{2n}, \text{ and}\\
\label{eq:sfb-cond-3}
	2^{-\sfb/2} & \le \frac{\prod_{j\not= i}\abs{z_i - z_j}}{16(n+1)2^{\tau_p}M(z_{i})^{n}}, 
\end{align}
then 
the disk $D(z_i, 2^{-\sfb/2})$ contains exactly one root approximation. 
For $i \neq j$, let $\hz_i$ and $\hz_{j}$ be arbitrary approximations
of $z_i$ and $z_j$ in the disks $D(z_i, 2^{-\sfb/2})$ and $D(z_j, 2^{-\sfb/2})$, respectively. Then,
\begin{equation}
	\label{eq:z-hz-sep}
   \Big(1 - \frac{1}{n}\Big)\cdot \abs{z_i - z_j} \le \abs{\hz_i - \hz_j} \le \Big(1 + \frac{1}{n}\Big)\cdot  \abs{z_i - z_j}. 
\end{equation}
\end{lemma}

Based on the previous lemma and Pan's root approximation algorithm \cite{pan-rootfinding-jsc}, Mehlhorn, Sagraloff, and Wang \cite{msw-aprox-fact-15}
developed an algorithm 
for isolating and approximating the roots of a univariate polynomial
up to any desired precision. The algorithm has the additional capability 
to consider polynomials with bitstream coefficients.
A simplified version of their theorem \cite[Theorem~4]{msw-aprox-fact-15}
that considers 
square-free polynomials is as follows:

\begin{theorem}
	\label{thm:root-isol-approx}
	Consider a square-freee polynomial $p(x) = \sum_{i=0}^{n}{p_i x^i} \in \CC[x]$
	such that $\tfrac{1}{4} \leq p_n \leq 1$. 
	If $z_i$ are the roots of $p$, then $P_i \coloneq \prod_{j\not= i}\abs{z_i - z_j}$, for $i \in [n]$. 
We can compute isolating discs $D(\hz_i, R_i)$
	with radius $R_i < 2^{-\kappa}$,
	for the roots $z_i$ of $p$,
	in a number of bit operations upper
	bounded by 
\[
		\sOB\Big(n^3 + n^2 \tau_p + n \sum_{i=1}^{n} \lg M(P_i^{-1})
		+ n \sum_{i=1}^{n} \lg M(\Delta_i^{-1})
		+ n \kappa\Big) .
	\]
	For this bound, we need rational approximation of the coefficients of $p$ 
	up to precision of $L$ bits, where
	\[
		L = \sOO\Big(n \Gamma_p  + \sum_{i=1}^{n} \lg M(P_i^{-1})
		+ \sum_{i=1}^{n} \lg M(\Delta_i^{-1})
		+ n \kappa\Big) .
	\]
The numbers $\hz_i$ are the approximations to the roots $z_i$
	and $R_i \leq \Delta_i/(64n)$.
\end{theorem}

The algorithm supported by Theorem~\ref{thm:root-isol-approx}
and its complexity depend on the geometry of the roots, that is the (aggregate) separation bound. Moreover, they do not depend on the type and the size of the coefficients.

\subsection{Bounds on the mimimum and the Fan-in algorithm}
\label{sec:fan-in}

Besides root approximations, we also need to bound the minimum of a
univariate polynomial, the evaluation of this polynomial at the roots
of its derivatives, and an algorithm to compute an approximation of
this polynomial from approximations of its roots.

\begin{lemma}
	\label{lem:eval-A-at-da}
	\label{lem:eval-dA-at-a}
	
	Consider the polynomial 
	\[A(x) = \sum_{k=0}^n a_k x^k = a_n \prod_{i=1}^n (x - \alpha_i) \in \QQ[x], \]
	such that $\tfrac{1}{2} \leq a_n \leq 1$
	and the other coefficients are rationals with the same denominator
	of bitsize at most $\tau$. 
	Let $A'$ be the derivative of $A$ with respect to $x$ 
	and let $\beta_j$, $j \in [n-1]$, be its roots. 
	Then, either $A'(\alpha_i) = 0$  (or $A(\beta_j) = 0$), or it holds
	\[
		2^{-4 n \tau - 16n\lg{n}}
		\leq \abs{A'(\alpha_i)}, \abs{A(\beta_j)}  \leq 
		2^{2 n \tau + 8n\lg{n}} 
		\]
        for all $i \in [n]$ and $j \in [n-1]$.
\end{lemma}
\begin{proof}
	The proof is a direct application of the bounds of the resultant,
	appeared in \cite{emt-dmm-j}.
	
	For the first bound, we consider the resultant 
    \[ H = \res(A(x), y - A'(x), x) = a_n^{n-1} \prod_{i=1}^{d}( y - A'(\alpha_i)) \in \QQ[y], \]
    that eliminates $x$; the last equality is due to the Poisson formula of the resultant.
   Then, $H$ is a univariate polynomial in $y$, and its roots are the evaluation of $A'$ at the roots of $A$.
    To bound the coefficients of $H$, proceed as follows.
    We notice that $H$ is a homogeneous polynomial of degree $n-1$ in the coefficients of $A$ and homogeneous of degree $n$ in the coefficients of $y - A'(x) \in (\QQ[y])[x]$.
	Specifically, $H$ is of the form 
	\[
		H = \dots + \varrho \, \bm{a}_{1}^{n-1} \bm{a}_{2}^{n} + \dots,
  	\]
	where $\varrho \in \ZZ$, $\bm{a}_1^{n-1}$ denotes a monomial in the coefficients of $A$ of total degree $n$,
	and $\bm{a}_2^{n-1}$ denotes a monomial in the coefficients of $y-A'$ of total degree $n$.

	We can bound $\varrho$ using \cite{sombra-mix-res-height}, see also \cite[Table~1 and Eq.~(1)]{emt-dmm-j}
	as 
	 \[\abs{\varrho} \leq (n+1)^{n-1} n^n \leq n^{2n} .\]
Since the bitsize of $A$ is at most $\tau$, we can upper bound $\bm{a}_1^{n-1} \in \QQ$ as 
	\[\abs{\bm{a}_1^{n-1}} \leq (2^{\tau})^{n-1} = 2^{\tau (n-1)};\]
	also the denominator of $\bm{a}_1^{n-1}$ is a integer at most  $2^{\tau (n-1)}$.
	
	To upper bound $\bm{a}_2^{n}$, we consider the worst case scenario
	that every coefficient of $y-A'$ is $y - n 2^{\tau}$.
	So $\bm{a}_2^{n-1}$ is a polynomial in $y$ 
	with coefficients rational number with numerator at most $\binom{n}{2} 2^{n \tau + n \lg{n}}$
	and denominator at most $2^{n \tau}$.
	
	Taking all these bounds into account, 
	$H$ is a polynomial in $\QQ[y]$ of degree at most $n$,
	its leading coefficient 	is in $[ (\tfrac{1}{2})^{n-1}, 1]$,
	and the other coefficients are rationals 
	having a numerator with magnitude at most $2^{2 n \tau + 7 n \lg{n}}$,
	and a denominator of magnitude at most $2^{2\tau(n-1)}$.
	
	Consequently, 
	using \cite[Theorem~1]{emt-dmm-j},
	we can bound the roots of $H$,
	and thus the evaluations $A'(\alpha_i)$, as follows
	\[
		 2^{-4 n \tau - 16n\lg{n}}
		 \leq 
		\frac{\tc(H)}{2 \normi{H}} 
		\leq \abs{A'(\alpha_i)} \leq 
		2 \frac{\normi{H}}{\lc(H)}
		\leq 
		2^{2 n \tau + 8n\lg{n}} .
	\] 	
	
	The same bounds hold for $A(\beta_j)$,
	where we use the resultant
	$\res(A'(x), y - A(x), x) \in \QQ[y]$.
\end{proof}

The following theorem supports an algorithm to compute an approximation 
of a polynomial $A$ when we have approximations of its roots.

\begin{lemma}[{\cite[Theorem~17]{pt-struct-mat-j}}]
  \label{lem:fan-in}
  Assume that we are given $n$ complex numbers $z_i$
  known up to an absolute precision $\lambda$, 
  that is, we know a rational $\hz_i$ such that 
  \[\abs{z_i - \hz_i} \leq 2^{-\lambda}.\]
  Also assume that  
  $\abs{z_i} \leq 2^{\tau}$, for a positive integer $\tau$.

  Based on the Fan-in process of the Moenck--Borodin algorithm~\cite{BorMoe-jcss-74}
  we can approximate the (coefficients of the) polynomial 
  with the (rational coefficients of the) polynomial \[\hatm(x) = \prod_i(x - \hz_i), \]
  so that it holds
  \[\normi{m - \hatm} \leq 2^{-\ell +(4n-4)\tau + 32n - (\lg{n} +5)^2 - 7} , \]
 at the cost of 
  $\sOB( n (\ell +n\tau) )$
  bit operations. 
Moreover,  it holds \[\lg \normi{m} \leq n\tau + 8n - 2\lg{n} - 8.\]
\end{lemma}

\subsection{About T-systems}
\label{sec:T-systems}

We borrow the following definition from \cite{Dio-T-systes-arxiv},
where we also refer the reader for further details.
For the following definition and further properties of the T-systems
see \cite[Chapter~4]{Dio-T-systes-arxiv}.
\begin{definition}[T-systems]
	\label{dfn:tSystem}
	Let $n\in \NN_0$, $\cX$ be a set with $|\cX|\geq n+1$,
	and $\cF = \{f_i\}_{i=0}^n$ be a family of real functions 
	$f_i:\cX\to\RR$. 
	In this setting, a polynomial the following linear combination 
	\begin{equation} 
		\label{eq:linF}
		f = \sum_{i=0}^n a_i\cdot f_i \quad
		\in\lin(\cF) := \{a_0 f_0 + \dots + a_n f_n \,|\, a_0,\dots,a_n\in\RR\}
			.
\end{equation}
The family $\cF$ on $\cX$ is a \emph{Tchebycheff system}
(or short \emph{T-system})
 \emph{of order $n$ on $\cX$} if every polynomial $f\in\lin(\cF)$ with 
 $\sum_{i=0}^n a_i^2 > 0$ has at most $n$ zeros in $\cX$.

If additionally $\cX$ is a topological space and $\cF$ is a family of continuous functions, then we call $\cF$ a \emph{continuous T-system}.
\end{definition}

The following theorem, due to Karlin, see  \cite[Chapter~7, Theorem~7.1]{Dio-T-systes-arxiv} and references therein,
leads to Positivstellensatz 
  for positive univariate polynomials.

\begin{theorem}
	\label{thm:Karlin}
Let $n\in\NN_0$, $\cF = \{f_i\}_{i=0}^n$ be a continuous T-system of order $n$ on $[a,b]$ with $a<b$, and let $f\in \cC([a,b],\RR)$ with $f>0$ on $[a,b]$ be a strictly positive continuous function. The following hold:
\begin{enumerate}[(i)]
\item There exists a unique polynomial $f_*\in\lin(\cF)$ such that
\begin{enumerate}[(a)]
\item $f(x) \geq f_*(x) \geq 0$ for all $x\in [a,b]$,
\item $f_*$ vanishes on a set with index $n$,
\item the function $f-f_*$ vanishes at least once between each pair of adjacent zeros of $f_*$,
\item the function $f-f_*$ vanishes at least once between the larges zero of $f_*$ and the end point $b$, and
\item $f_*(b)>0$.
\end{enumerate}
\item There exists a unique polynomial $f^*\in \lin(\cF)$ which satisfies the conditions \mbox{(a) to (d)} of (i) and
\begin{enumerate}[(a')]\setcounter{enumii}{4}
\item $f^*(b) = 0$.
\end{enumerate}
\end{enumerate}
\end{theorem}

As consequences of \Cref{thm:Karlin} we get the following (much sharper versions of known) results,
where we emphasize the uniqueness of $\alpha,\beta>0$, $x_i$, and $y_i$.

\begin{corollary} \label{cor:tsystemsAllLine}
    Let $p\in\RR[x]$ with even degree $\deg p = 2m$, $m\in\NN$, and $p\geq 0$ on $\RR$.
    Let $z_1,\dots,z_k$ be the zeros of $p$ in $\RR$ with (even) multiplicities $m_1,\dots,m_k\in 2\NN$.
    Then there exist unique $\alpha,\beta>0$ and unique
\[x_1 < y_1 < x_2 < \dots < y_{l-1} < x_l\]
with $\deg p = m_1 + m_2 + \dots + m_k + 2l$ such that
\[p(x) = \prod_{i=1}^k (x-z_i)^{m_i} \cdot \left(\alpha\cdot \prod_{i=1}^l (x-x_i)^2 + \beta\cdot \prod_{i=1}^{l-1} (x-y_i)^2 \right).\]
\end{corollary}

In this setting, we consider polynomials in $\RR[x]$,
thus we can assume that we can assume that we have factored out their real zeros
and we consider only polynomials without real roots.
Hence, we state the following two consequences of \Cref{thm:Karlin} only for $p>0$.

\begin{corollary}
    Let $p\in\RR[x]$ with $p>0$ on $[0,\infty)$.
\begin{enumerate}[(i)]
        \item If $\deg p = 2m$ with $m\in\NN_0$, then there exist unique $\alpha,\beta>0$ and unique
\[0 < x_1 < y_1 < \dots < y_{m-1} < x_m < \infty\]
such that
\[p(x) = \alpha\cdot \prod_{i=1}^m (x-x_i)^2 + \beta\cdot x\cdot\prod_{i=1}^{m-1} (x-y_i)^2.\]

        \item If $\deg p = 2m+1$ with $m\in\NN_0$, then there exist unique $\alpha,\beta>0$ and unique
\[0 < x_1 < y_1 < \dots < x_m < y_m < \infty\]
such that
\[p(x) = \alpha\cdot\prod_{i=1}^m (x-x_i)^2 + \beta\cdot x\cdot\prod_{i=1}^m (x-y_i)^2.\]
    \end{enumerate}
\end{corollary}

\begin{corollary}
    Let $p\in\RR[x]$ with $p>0$ on $[a,b]$ for some $a,b\in\RR$ with $a < b$.
\begin{enumerate}[(i)]
        \item If $\deg p = 2m$ for some $m\in\NN_0$, then there exist unique $\alpha,\beta>0$ and unique
\[a< x_1 < y_1 < \dots < y_{m-1} < x_m < b\]
such that
\[p(x) = \alpha\cdot\prod_{i=1}^m (x-x_i)^2 + \beta\cdot (x-a)\cdot (b-x)\cdot \prod_{i=1}^{m-1} (x-y_i)^2.\]

        \item If $\deg p = 2m+1$ for some $m\in\NN_0$, then there exist unique $\alpha,\beta>0$ and
\[a < x_1 < y_1 < \dots < x_m < y_m < b\]
such that
\[p(x) = \alpha\cdot (x-a)\cdot\prod_{i=1}^m (x-x_i)^2 + \beta\cdot (b-x)\cdot\prod_{i=1}^m (x-y_i)^2.\]
    \end{enumerate}
\end{corollary}

Even though we have presented the statements that if $p>0$ (on $\RR$, $[0,\infty)$, or $[a,b]$), then there exists a unique representation as a sum of squares (with additional factors $x$, $x-a$, and $b-x$),
the opposite implications are straightforward. 
Thus, all the previous imply $p>0$. 
Therefore, these equivalences are positivity certificates.

\end{document}